\documentclass[a4paper,11pt,english]{article}
\pdfoutput=1

\usepackage{ifdraft}
\ifdraft{\newcommand{\authnote}[3]{{\color{#3} {\bf  #1:} #2}}}{\newcommand{\authnote}[3]{}}

\newcommand{\anote}[1]{\authnote{ Andr\'{a}s}{#1}{blue}}

\newif\ifcount
\ifdraft{\counttrue}

\usepackage[utf8]{inputenc}
\usepackage[T1]{fontenc}
\usepackage[left=1in,top=1in,right=1in,bottom=1in]{geometry}
\usepackage[english]{babel}
\usepackage{verbatim}

\usepackage{tikz,pgfplots}
\usetikzlibrary{backgrounds,fit,decorations.pathreplacing,calc}
\usepackage{amssymb}
\usepackage{mathtools}
\usepackage{bbm}
\usepackage{mleftright}\mleftright
\usepackage{multirow}
\usepackage{amsthm} 
\theoremstyle{plain}
\usepackage{thmtools} 
\usepackage{thm-restate}

\usepackage{algorithm}
\usepackage{algcompatible}
\usepackage{float}
\usepackage{enumitem}
\makeatletter
\def\namedlabel#1#2{\begingroup
#2%
\def\@currentlabel{#2}%
\phantomsection\label{#1}\endgroup
}

\usepackage{caption}
\usepackage{subcaption}

\usepackage{titling}
\usepackage[final=true,colorlinks = true,allcolors = {blue},]{hyperref}
\usepackage[capitalize]{cleveref}

\newcommand{\eps}{\varepsilon}

\newcommand{\ketbra}[2]{|#1\rangle\! \langle #2|}

\newcommand{\nrm}[1]{\left\lVert #1 \right\rVert}
\newcommand{\bigO}[1]{\mathcal{O}\left( #1 \right)}

\newcommand{\bigOt}[1]{\widetilde{\mathcal{O}}\left( #1 \right)}

\newcommand{\diag}[1]{\mathrm{diag}\left( #1 \right)}

\newcommand{\vertiii}[1]{{\left\vert\kern-0.25ex\left\vert\kern-0.25ex\left\vert #1 
\right\vert\kern-0.25ex\right\vert\kern-0.25ex\right\vert}}

\newcommand{\checkv}{\mathtt{Check}} 
\newcommand{\setupv}{\mathtt{Setup}}
\newcommand{\updatev}{\mathtt{Update}}
\newcommand{\ccheck}{\mathcal{C}} 
\newcommand{\csetup}{\mathcal{S}}
\newcommand{\cupdate}{\mathcal{U}}
\renewcommand{\check}{\mathsf{C}} 
\newcommand{\setup}{\mathsf{S}}
\newcommand{\update}{\mathsf{U}}
\newcommand{\refl}{\mathsf{R}}

\newcommand{\HT}{\mathrm{HT}}
\newcommand{\supp}{\mathrm{supp}}
\newcommand{\hit}{\mathrm{ht}}
\newcommand{\ct}{\mathrm{ct}}
\theoremstyle{definition}
\pgfplotsset{compat=1.13}

\newcommand{\N}{\mathbb{N}}
\newcommand{\R}{\mathbb{R}}

\DeclarePairedDelimiter\bra{\langle}{\rvert}
\DeclarePairedDelimiter\ket{\lvert}{\rangle}
\DeclarePairedDelimiterX\braket[2]{\langle}{\rangle}{#1 \delimsize\vert #2}

\def\Pr{\mathrm{Pr}}
\def\Exp{\mathbb{E}}

\newtheorem{theorem}{Theorem}
\newtheorem{corollary}[theorem]{Corollary}
\newtheorem{lemma}[theorem]{Lemma}
\newtheorem{definition}[theorem]{Definition}
\newtheorem{claim}[theorem]{Claim}
\newtheorem{remark}[theorem]{Remark}

\title{A Unified Framework of Quantum Walk Search
}
\author{
Simon Apers\thanks{Inria, France and CWI, the Netherlands. Supported by the CWI-Inria International Lab. \texttt{simon.apers@inria.fr}}
\and
András Gilyén\thanks{QuSoft, CWI and University of Amsterdam, the Netherlands. Supported by ERC Consolidator Grant QPROGRESS and partially supported by QuantERA project QuantAlgo 680-91-034.}
\kern1.3mm$^,$\thanks{Caltech, USA. Supported by the Institute for Quantum Information and Matter, an NSF Physics Frontiers Center (NSF Grant PHY-1733907). \texttt{agilyen@caltech.edu}}
\and
Stacey Jeffery\thanks{QuSoft and CWI, the Netherlands. Supported by an NWO Veni Innovational Research Grant under project number 639.021.752, an NWO WISE Grant,    and QuantERA  project  QuantAlgo  680-91-03.   SJ  is  a  CIFAR  Fellow  in  the Quantum Information Science Program. \texttt{jeffery@cwi.nl}}
}
\date{\today\vspace{-5mm}}

\begin{document}

\maketitle

\begin{abstract}
The main results on quantum walk search are scattered over different, incomparable frameworks, most notably the \emph{hitting time framework}, originally by Szegedy, the \emph{electric network framework} by Belovs, and the \emph{MNRS framework} by  Magniez, Nayak, Roland and Santha.
As a result, a number of pieces are currently missing.
For instance, the electric network framework allows quantum walks to start from an arbitrary initial state, but it only detects marked elements.
In recent work by Ambainis et al., this problem was resolved for the more restricted hitting time framework, in which quantum walks must start from the stationary distribution.

We present a new quantum walk search framework that unifies and strengthens these frameworks.
This leads to a number of new results.
For instance, the new framework not only detects, but finds marked elements in the electric network setting.
The new framework also allows one to interpolate between the hitting time framework, which minimizes the number of walk steps, and the MNRS framework, which minimizes the number of times elements are checked for being marked.
This allows for a more natural tradeoff between resources.
Whereas the original frameworks only rely on quantum walks and phase estimation, our new algorithm makes use of a technique called \emph{quantum fast-forwarding}, similar to the recent results by Ambainis et al.
As a final result we show how in certain cases we can simplify this more involved algorithm to merely applying the quantum walk operator some number of times.
This answers an open question of Ambainis et al.
\end{abstract}

\section{Introduction}

Quantum walk search refers to the use of quantum walks to solve a search problem on a graph.
In the last two decades, this topic has received a great deal of attention, with a rich literature attesting to the progress on understanding quantum walk algorithmic techniques~\cite{ambainis2005coins,szegedy2004QMarkovChainSearch,magniez2006SearchQuantumWalk,krovi2010QWalkFindMarkedAnyGraph,belovs2013Electric,ambainis2019QW,dohotaru2017controlledQAmp}
and developing applications~\cite{buhrman2006MatrixProduct,magniez2007Triangle,jeffery2012NestedQW,belovs20133Dist,bernstein2013SubsetSum,montanaro2015quantum,kachigar2017InformationSetDec,helm2018SubsetSum,kirshanova2018InformationSetDec}.
Despite this long line of progress, the main results on quantum walk search lie somewhat scattered in different frameworks, and a number of pieces are currently missing.

The quantum walk search frameworks that we consider are the \emph{hitting time framework} originally due to Szegedy~\cite{szegedy2004QMarkovChainSearch}, the \emph{MNRS framework} due to Magniez, Nayak, Roland and Santha~\cite{magniez2006SearchQuantumWalk}, the \emph{electric network framework} due to Belovs~\cite{belovs2013Electric}, and the \emph{controlled quantum amplification framework} by Dohotaru and H{\o}yer~\cite{dohotaru2017controlledQAmp}.
We summarize these frameworks, as well as the corresponding complexities, in Table~\ref{fig:frameworks}.
In this work we unify these different frameworks, leading to a number of new results and missing pieces.
For example, algorithms developed using the electric network framework could only \emph{detect} marked elements.
Our unified approach can be used to develop algorithms that \emph{find} marked elements, while incurring at most a logarithmic overhead.

We also give a conceptual bridge between the recent result of Ref.~\cite{ambainis2019QW} and the original approaches by Szegedy~\cite{szegedy2004QMarkovChainSearch} and Krovi, Magniez, Ozols and Roland~\cite{krovi2010QWalkFindMarkedAnyGraph}.
The latter showed that combining quantum walks with phase estimation or time averaging allows one to quadratically improve the hitting time of a single marked element, when starting from the stationary distribution.
Ambainis et al.~\cite{ambainis2019QW} used a more involved technique called \emph{quantum fast-forwarding}~\cite{apers2018QFastForwardMarkovChains} to improve these results to yield quadratic speedups on the hitting time of arbitrary sets.
In this work we reprove the same result using only simple quantum walks, thereby proving a conjecture from~\cite{ambainis2019QW}.

\subsection{Different Frameworks}
While the frameworks we consider are similar, each has advantages and disadvantages.
The earliest \emph{hitting time framework} was due to Szegedy~\cite{szegedy2004QMarkovChainSearch}, inspired by an algorithm of Ambainis for element distinctness~\cite{ambainis2007ElementDist}.
To illustrate this framework, imagine a classical algorithm that begins by sampling a state from the stationary distribution $\pi$ of some random walk, described by a transition matrix $P$.
The algorithm starts from a vertex distributed according to $\pi$, and simulates the random walk.
After every step of the walk it checks whether the current vertex is ``marked''.
The algorithm terminates after $\bigO{\HT(P,M)}$ steps, with $\HT(P,M)$ the \emph{hitting time}, or the expected number of steps from $\pi$ before a marked vertex in $M$, the marked set, is reached.
As such, the algorithm has a constant probability of having found a marked vertex.
To bound the complexity of this algorithm, let the \emph{setup cost} $\csetup$ denote the complexity of sampling from $\pi$, the \emph{update cost} $\cupdate$ denote the complexity of simulating a step of the walk, and the \emph{checking cost} $\ccheck$ denote the complexity of checking whether a vertex is marked.
The complexity of the resulting algorithm is then of order $\csetup+{\HT(P,M)}(\cupdate+\ccheck)$.
The hitting time framework essentially shows how to construct a quantum algorithm with complexity
\[
\setup + \sqrt{\HT(P,M)}(\update + \check),
\]
where $\setup$, $\update$ and $\check$ are quantum analogues of $\csetup$, $\cupdate$ and $\ccheck$, respectively, denoting the costs in terms of \emph{coherent} quantum samples (see Section~\ref{sec:prelim-qw-search} for details).
One of the major drawbacks of the original framework was that the resulting quantum algorithm typically \emph{detected} the presence of a marked vertex, without actually \emph{finding} one.
In the special case where there is only a single marked element, Krovi at al.~\cite{krovi2010QWalkFindMarkedAnyGraph} showed how to also find the marked element in the same complexity.
To this end they introduced the concept of \emph{interpolated walks}.
Combining interpolated walks with another technique called \emph{quantum fast-forwarding}, introduced in \cite{apers2018QFastForwardMarkovChains}, Ref.~\cite{ambainis2019QW} more recently showed how to also find a marked element in the general case.
We will refer to this final result as the hitting time framework.

The second framework that we consider is the \emph{MNRS framework} introduced by Magniez, Nayak, Roland and Santha~\cite{magniez2006SearchQuantumWalk}.
This framework also \emph{finds} a marked vertex, but it can be understood as the quantum analogue of a slightly different random walk algorithm.
Consider a random walk that begins in the stationary distribution.
Rather than checking if the current vertex is marked after every step, the walk takes $1/\delta$ steps between checks, where $\delta$ is the \emph{spectral gap} of $P$.
Since $1/\delta$ is approximately the \emph{mixing time} of the random walk, this process effectively samples from the stationary distribution, for each sample checking whether it is marked, and otherwise generating a new sample.
If $\eps$ is the probability that a vertex sampled from the stationary distribution is marked, then a marked element is found with constant probability after $\bigO{1/\eps}$ samples.
As such, the complexity of this classical algorithm is $\csetup+\frac{1}{\eps}(\frac{1}{\delta}\cupdate+\ccheck)$.
The MNRS framework shows how to get a quantum algorithm for finding a marked vertex with complexity
\[
\setup + \frac{1}{\sqrt{\eps}}\Big(\frac{1}{\sqrt{\delta}}\update + \check\Big).
\]
Since $\HT(P,M) \leq \frac{1}{\eps\delta}$, this requires at least as many steps of the walk as the hitting time framework.
On the other hand, $\HT(P,M) \geq \frac{1}{\eps}$, and so the number of checks can be significantly smaller than in the hitting time framework.
In fact, this amount of checks performed in the MNRS framework is easily seen to be optimal by a lower bound on black-box search.\footnote{Consider for instance a quantum walk search algorithm on the complete graph on $N$ vertices. Finding a single marked element then requires $\Omega(\sqrt{N})$ checks by the optimality of Grover's search algorithm.}

The third framework that we consider is the \emph{electric network framework} by Belovs~\cite{belovs2013Electric} (published in~\cite{belovs20133Dist}).
This is a generalization of the hitting time framework, allowing for the walker to start from an arbitrary initial distribution $\sigma$ (such as a single vertex), rather than necessarily the stationary distribution.
If $\setup(\sigma)$ is the complexity of sampling (coherently) from $\sigma$, then the resulting quantum algorithm has complexity
\[
\setup(\sigma)+\sqrt{C_{\sigma,M}}(\update(\sigma)+\check),
\]
where $\update(\sigma)$ is the complexity of implementing a step of a slightly modified random walk.
The quantity $C_{\sigma,M}$ (defined in Section~\ref{sec:prelim-electric}) is a generalization of the \emph{commute time}. 
When both $\sigma$ and $M$ correspond to single vertices $u$ and $m$, then $C_{\sigma,M}$ equals the commute time from $u$ to $m$, which is the expected number of steps starting from $u$ to reach $m$ and then return to $u$. 
When $\sigma$ equals the stationary distribution then $C_{\sigma,M} = HT(P,M)$, thus retrieving the hitting time framework.
The obvious advantage of the electric network framework is that it does not necessarily require quantum samples from the stationary distribution of $P$, which might be very costly, and can instead begin in a much easier to produce state. 
A major disadvantage of this framework, however, is that the quantum algorithm only \emph{detects} the presence of marked vertices, as in the original hitting time framework, rather than actually finding marked vertices.

Finally we also consider the \emph{controlled quantum amplification framework} by Dohotaru and H{\o}yer \cite{dohotaru2017controlledQAmp}.
They use an extra qubit to control the quantum walk operator\footnote{In fact they consider more general operators, but we will focus on their result for quantum walk operators.}, leading to an additional degree of freedom.
For the case of a \emph{unique} marked element $M = \{m\}$, and starting from a quantum sample of the stationary distribution, they achieve a complexity
\[
\setup+\sqrt{\HT(P,\{m\})}\update+\frac{1}{\sqrt{\eps}}\check,
\]
which has both an optimal number of walk steps (as the hitting time framework) and an optimal number of checks (as the MNRS framework).
The clear downside of this approach is that it is restricted to cases where there is a single marked element, and we start from the stationary distribution.

\begin{table}
\centering
\renewcommand{\arraystretch}{1.3}
\begin{tabular}{|l|l|}
\hline
Framework & Complexity\\
\hline
Hitting time framework \cite{szegedy2004QMarkovChainSearch,krovi2010QWalkFindMarkedAnyGraph,ambainis2019QW} & $\mathsf{S}+\sqrt{\HT(P,M)}(\mathsf{U}+\mathsf{C})$\\
\hline
MNRS framework \cite{magniez2006SearchQuantumWalk} & $\mathsf{S}+\frac{1}{\sqrt{\eps}}(\frac{1}{\sqrt{\delta}}\mathsf{U}+\mathsf{C})$\\
\hline
Electric network framework \cite{belovs20133Dist,belovs2013Electric} & $\mathsf{S}(\sigma)+\sqrt{C_{\sigma,M}}(\mathsf{U}(\sigma)+\mathsf{C})$\\
\hline
Controlled quantum amplification \cite{dohotaru2017controlledQAmp} & $\mathsf{S}+\sqrt{\HT(P,\{m\})}\mathsf{U}+\frac{1}{\sqrt{\eps}}\mathsf{C}$\\
\hline
\end{tabular}
\caption{Comparison of different quantum walk frameworks.}\label{fig:frameworks}
\end{table}

\subsection{Contributions}

\paragraph{Finding in the electric network framework}

The electric network framework \cite{belovs2013Electric} generalizes the hitting time framework~\cite{szegedy2004QMarkovChainSearch} by allowing for arbitrary initial distributions.
The downside is that algorithms in this framework only detect rather than actually find marked vertices.
On the other hand, the improved hitting time framework of~\cite{ambainis2019QW} shows how to actually find marked vertices in the hitting time framework, provided that the walk starts from a quantum sample of the stationary distribution.
Both works hence provide complementary but incompatible improvements over the initial hitting time framework.

In Section~\ref{sec:electric}, we fill this gap by generalizing the results of~\cite{ambainis2019QW} to the electric network setting, designing a quantum algorithm that not only detects but also \emph{finds} marked elements for any starting distribution $\sigma$.
This improved version strictly generalizes the results of~\cite{ambainis2019QW}, and it loses at most a log factor with respect to the original electric network framework~\cite{belovs2013Electric}. In particular, we show (see Theorem~\ref{thm:electric-finding}):

\begin{theorem}[Informal]\label{thm:intro-elec}
For any distribution $\sigma$, there is a quantum walk search algorithm that finds a marked element from $M$ with constant probability in complexity (up to log factors)
\[
\setup(\sigma) + \sqrt{C_{\sigma,M}} (\update(\sigma)+\check).
\]
\end{theorem}

To analyze our new algorithm, we use techniques similar to those employed in~\cite{ambainis2019QW} for finding in the hitting time framework.
However, there is an additional difficulty we must overcome.
The hitting time, $\HT(P,M)$, has a useful interpretation in terms of the classical random walk -- that is, with high probability, a marked vertex is encountered within the first $\bigO{\HT(P,M)}$ steps -- and this fact is crucial in the analysis of the quantum algorithm in~\cite{ambainis2019QW}.
In contrast, to the best of our knowledge, the generalized quantity $C_{\sigma,M}$ (defined in Section~\ref{sec:prelim-electric}) is not well understood.
If $\sigma$ is supported on a single vertex, $u$, and $M$ contains a single vertex, $m$, then $C_{\sigma,M}$ is exactly the \emph{commute time} between $u$ and $m$.
This means that within the first $\bigO{C_{\sigma,M}}$ steps, with high probability, a walker starting from $u$ has visited $m$, and then returned to $u$.
For general $\sigma$ and $M$, no such interpretation was known.
We prove that, under certain conditions, a similar interpretation holds: with high probability, a walker starting from $\sigma$ will hit $M$, and then return to the support of $\sigma$, within the first $\bigO{C_{\sigma,M}}$ steps.
We can ensure that these conditions hold by using the same graph and walk modification as used in \cite{belovs2013Electric}, adding a weighted edge to each vertex in $\mathrm{supp}(\sigma)$.
The resulting understanding of $C_{\sigma,M}$ is enough to employ a similar analysis to that of~\cite{ambainis2019QW}.

\paragraph{A Unified Framework} While the electric network framework is a generalization of the hitting time framework, the MNRS framework is incomparable.
Since $\HT(P,M) \leq \frac{1}{\eps\delta}$, the hitting time framework always finds a marked vertex using a number of quantum walk steps (updates) less than or equal to that used by the MNRS framework.
On the other hand, $\HT(P,M) \geq \frac{1}{\eps}$, and hence the MNRS framework may make fewer calls to the check operation.
When the complexity of implementing the checking operation is much larger than that of the update operation, the MNRS framework may hence be preferable to both the hitting time framework and the electric framework.
The controlled quantum amplification framework achieves the best of both worlds, but only for a unique marked element.

In Section~\ref{sec:unified}, we present a new framework that unifies all these individual approaches.
For the sake of intuition, we first describe this framework when the initial state $\pi$ is used, which can be seen as a unification between the hitting time framework, the MNRS framework and the controlled quantum amplification framework.
To this end, recall that the hitting time framework is the quantum analogue of a random walk algorithm that takes $\HT(P,M)$ steps of the random walk described by $P$, checking at each step if the current vertex is marked.
In contrast, the MNRS framework is the quantum analogue of a random walk algorithm that takes $\frac{1}{\delta}$ steps of $P$, thus approximately sampling from the stationary distribution $\pi$, and then checks if the sampled vertex is marked.
Since $\eps$ is the probability that a sampled vertex is marked, this process is repeated $\frac{1}{\eps}$ times.

We can define a natural interpolation between both classical algorithm.
To this end, take any $t$, and consider a classical random walk that repeatedly takes $t$ steps, and then checks whether the current vertex is marked.
The expected number of iterations is then $\HT(P^t,M)$, the hitting time of the $t$-step random walk, described by transition matrix $P^t$.
This classical algorithm finds a marked vertex in complexity $\csetup + \HT(P^t,M)(t\cupdate + \ccheck)$.
We give a quantum analogue of this algorithm, generalized to arbitrary initial distributions (see Theorem~\ref{thm:mnrs-electric}).
\begin{theorem}[Informal] \label{thm:prelim-QW}
For any $t \in \N$ and any distribution $\sigma$, there is a quantum walk search algorithm that finds a marked element from $M$ with constant probability in complexity (up to log factors)
\[
\setup(\sigma) + \sqrt{C_{\sigma,M}(P^t)} (\sqrt{t}\update(\sigma)  + \check).
\]
\end{theorem}
Setting $t = 1$ we recover our previous theorem, Theorem~\ref{thm:intro-elec}.
When $\sigma = \pi$, then $C_{\sigma,M}(P^t) = \HT(P^t,M)$, and hence we find the quantum analogue of the aforementioned random walk algorithm.
As such, when $\sigma = \pi$ and $t = 1$, we recover the hitting time framework.
When $\sigma = \pi$ and $t=\frac{1}{\delta}$, we recover the MNRS framework, since a $1/\delta$-step random walk essentially samples from $\pi$ at every step, and so $\HT(P^{1/\delta},M) \in \bigO{\frac{1}{\eps}}$.
When there is a unique marked element $\{m\}$, and $t = \eps \HT(P,\{m\})$, we recover the controlled quantum amplification framework.
To see this, we use a result from \cite[Section 6]{dohotaru2017controlledQAmp} which proves that $\HT(P^t,\{m\}) = 1/\eps$ if $t \in \Omega(\eps \HT(P,\{m\}))$.
For multiple marked elements, and other intermediate values of $t$, we obtain new types of algorithms. 
We summarize these special cases in the table below.

\begin{table}[H]
\centering
\renewcommand{\arraystretch}{1.2}
\begin{tabular}{|l|l|}
\hline
New quantum walk search framework: & $\mathsf{S}(\sigma)+\sqrt{C_{\sigma,M}(P^t)}(\sqrt{t}\mathsf{U}(\sigma)+\mathsf{C})$\\
\hline
Hitting time framework & $\sigma=\pi$, $t=1$\\
MNRS framework & $\sigma=\pi$, $t=\frac{1}{\delta}$\\
Electric network framework & any $\sigma$, $t=1$\\
Controlled quantum amplification & $\sigma=\pi$, $M = \{m\}$, $t=\eps \HT(P,M)$\\
\hline
\end{tabular}
\caption{The new quantum walk search framework.}\label{fig:new-framework}
\end{table}

\paragraph{A simpler algorithm for the hitting time and electric network framework}
Similar to the recent work by Ambainis et al.~\cite{ambainis2019QW}, our new quantum algorithm makes use of a somewhat involved technique called quantum fast-forwarding.
For the case $t = 1$ (recovering the hitting time and electric network framework), we show that a much simpler algorithm works with essentially the same complexity.
This algorithm works by (classically) choosing random interpolation parameters, and applies the interpolated quantum walk operator an appropriately chosen number of steps, starting from $\ket{\sqrt{\sigma}}$.
It was already conjectured in \cite{ambainis2019QW} that this simpler approach would work (but only for the hitting time framework).
In Section~\ref{sec:alter}, we prove that this simple algorithm indeed finds a marked vertex, with at most a logarithmic overhead over the complexity of the more involved fast-forwarding algorithm.
Interestingly, our proof relies on the proof of correctness of the fast-forwarding algorithm.

\paragraph{Related independent work} While finalizing this manuscript, the authors became aware of the concurrent and independent work of Stephen Piddock, who developed an alternative refinement of Belovs' results for finding marked elements in the electric network framework \cite{piddock2019electricfind}.

\section{Preliminaries} \label{sec:prelim}

\subsection{Random walks} \label{sec:prelim-rws}

Let $Y$ be a random variable over a finite state space $X$, with $|X| = n$.
For $y \in X$, we let $\Pr(Y = y)$ denote the probability that $Y = y$.
We can describe the corresponding probability distribution by a vector $\sigma \in \R^n$, where $\sigma_y = \Pr(Y = y)$.
For $S \subseteq X$ and a probability distribution $\sigma$, we let $\sigma(S)=\sum_{u\in S} \sigma_u$ denote the probability that $Y \in S$, and we let $\sigma|_S$ be the normalized restriction of $\sigma$ to $S$, defined as $(\sigma|_S)_u=\sigma_u/\sigma(S)$ for all $u\in S$.
A sequence of random variables $Y=(Y_t)_{t=0}^{\infty}$ over $X$ is a \emph{Markov chain} if for all $t\geq 1$, $Y_t$ is independent of $Y_0,\dots,Y_{t-2}$ given $Y_{t-1}$.
For any distribution $\sigma$ over $X$ and random variable $Z$ that is a function of $Y$, we let $\Pr_\sigma(Z=z)$ denote the probability that $Z$ takes value $z$ when $Y_0$ is distributed as $\sigma$.
Any Markov chain is described by a stochastic transition matrix $P$, with $P_{y,y'}=\Pr(Y_t = y' \mid Y_{t-1} = y)$.
If $\sigma^{(t)}$ describes the probability distribution of $Y_t$, then this implies that $\sigma^{(t)} = \sigma^{(t-1)} P$.

We consider weighted, undirected graphs $G = (X,E,w)$ on vertex set $X$, with $|X| = n$; edge set $E \subseteq X \times X$ with $(u,v)\in E$ if and only if $(v,u)\in E$, and $|E| = 2m$; and edge weights $w:E \to \R_{\geq 0}$.
If an edge is not present, we will usually think of it as having edge weight zero.
The total weight is then $W = \sum_{u,v\in X} w_{u,v}$, and the total weight of edges leaving a node $u$ is $w_u = \sum_{v \in X} w_{u,v}$.
A random walk on $G$ is described by a Markov chain over $X$ with transition matrix $P$, defined by
\[
P_{u,v}
= \frac{w_{u,v}}{w_u}.
\]
In words, a random walk from a vertex $u$ picks a random neighbor $v$ with probability proportional to the edge weight $w_{u,v}$.
This random walk describes a special kind of Markov chain, a so-called \emph{reversible} Markov chain \cite{levin2017MarkovChainsMixingTimes}.
In fact, it can be shown that any reversible Markov chain can also be described as a random walk on a weighted graph, simply by choosing $w_{u,v} = w_{v,u} = 2\pi_uP_{u,v}$.
As such we will interchangeably use the terms ``random walk'' and ``reversible Markov chain''.
If the graph is connected and non-bipartite, the random walk is called \textit{ergodic} and its probability distribution converges to a unique limiting distribution called the \textit{stationary distribution} $\pi \in \R^n$, defined as
\[
\pi_u
= \frac{w_u}{W}.
\]
This is the unique left eigenvector of $P$ with eigenvalue 1.
While the transition matrix $P$ of a reversible Markov chain is not necessarily symmetric, a closely related matrix called the \textit{discriminant matrix} $D(P)$ is symmetric.
For ergodic and reversible Markov chains it can be defined as
\begin{equation} \label{eq:discriminant}
D(P):=\sqrt{P \circ P^T}
= \diag{\sqrt{\pi}} P \,\diag{\sqrt{\pi}}^{\!-1},
\end{equation}
with the $\circ$-product and square root acting elementwise and
$\diag{\sqrt{\pi}}$ being the diagonal matrix with entries $(\diag{\sqrt{\pi}})_{u,u} = \sqrt{\pi_u}$.
The second equality implies that $P$ and $D(P)$ share the same eigenvalues, which we denote by $1 = \lambda_0 > \lambda_1 \geq \dots \geq \lambda_{n-1} > -1$.
The convergence time or \textit{mixing time} of $P$ to the stationary distribution is characterized, up to log factors, by the inverse of the \textit{spectral gap} $\delta$ of the transition matrix, which is defined as $\delta = \min\{1-|\lambda_1|,1-|\lambda_{n-1}|\}$.

For a subset $M \subseteq X$, the \textit{hitting time} $\tau_M$ is the random variable representing the minimum $t$ such that $Y_t\in M$, i.e., the first time at which the random walk hits $M$.
With slight abuse of notation, we will also call $\HT(P,M) = \Exp_\pi(\tau_M)$ the hitting time of set $M$, corresponding to the expected hitting time when starting from the stationary distribution $\pi$.
We similarly define $\tau^M_S$ to be the first time at which the random walk has hit $M$, and then $S$.
As such, when $S = \{s\}$ is a singleton, the quantity $\Exp_s(\tau^M_s)$ denotes the expected \textit{commute time} from $s$ to $M$.

\subsection{Interpolated walks} \label{sec:prelim-interpolated}
\emph{Interpolated walks} form an important tool in quantum walk search algorithms~\cite{krovi2010QWalkFindMarkedAnyGraph}.
Quite literally, such walks are an interpolation between the original random walk $P$, and the \emph{absorbing} random walk $P_M$ in some set $M$, which corresponds to the walk that halts when it hits $M$ (equivalently, the walk obtained after adding self-loops of infinite weight to the elements in $M$).
For an interpolation parameter $s \in [0,1]$, the interpolated walk $P(s)$ is then defined as
\[
P(s)
= (1-s) P + s P_M.
\]
If $P$ is an ergodic reversible Markov chain, then so is $P(s)$ for every $s\in[0,1)$ \cite[Proposition 12]{krovi2010QWalkFindMarkedAnyGraph}.

\subsection{Electric networks} \label{sec:prelim-electric}
An electric network is described by a weighted, undirected graph $G=(X,E,w)$ (for an introduction to the connection between random walks on graphs and electric networks, see~\cite{doyle1984electric}).
The edge weights in $G$ are interpreted as \emph{conductances} associated to the edges.
An edge that is not present has weight zero, and hence also zero conductance.
A central notion is that of a \emph{flow}.

\begin{definition} \label{def:flow}
Let $M\subset X$ be a set of marked vertices, and let $\sigma$ be a distribution supported on unmarked vertices. A \emph{unit flow} from $\sigma$ to $M$ is a function $p:X\times X\rightarrow \mathbb{R}$ such that:
\begin{itemize}
	\item $p_{u,v}=0$ if $w_{u,v}=0$;
	\item $p_{u,v}=-p_{v,u}$ for all $u,v$;
	\item for all $u\not\in M$, $\sum_{v}p_{u,v}=\sigma_u$; and
	\item $\sum_{u\in M} \sum_{v\notin M} p_{u,v} = -1$.
\end{itemize}
The \emph{effective resistance} $R_{\sigma,M}$ from $\sigma$ to $M$ is 
\begin{align}\label{eq:resistance}
R_{\sigma,M}
= \min_{p}\sum_{u,v\in X:u<v}\frac{p_{u,v}^2}{w_{u,v}}
\end{align}
where the minimum runs over all unit flows from $\sigma$ to $M$. For an $S$ disjoint from $M$ we define
$$
	R_{M,S}=R_{S,M}:=\min_{\sigma\colon \mathrm{supp}(\sigma)\subseteq S}R_{\sigma,M}
$$
In the special case when $S=\{s\}$ and $M=\{t\}$ are singletons, we simply write $R_{s,t}:=R_{\{s\},\{t\}}$. 
\end{definition}
Note that there is a unique $\sigma$-$M$ flow that minimizes the expression in \eqref{eq:resistance}.
To see this, note that if two distinct flows $p$ and $p'$ achieve the same minimum, then their average is again a $\sigma$-$M$ flow, but with an even smaller value, which leads to a contradiction.
In the special cases where $\sigma = \pi$ or $\supp(\sigma) = \{s\}$, the effective resistance $R_{\sigma,M}$ has a well-known combinatorial interpretation.

\begin{theorem}[\cite{chandra1996electrical,belovs2013Electric}]\label{thm:commuteVertex}
In a weighted graph of total weight $W$, for any vertex $s$ and subset $M \subseteq X$, we have $WR_{s,M}=C_{s,M}$, where $C_{s,M}$ is the \emph{commute time} from $s$ to $M$.
Furthermore, we have $WR_{\pi,M}=\HT(P,M)$, with $\HT(P,M)$ the hitting time from $\pi$ to $M$.
\end{theorem}
\noindent
Since we could only find in the literature a proof of the first statement for the case where $M$ is a singleton, we extend the existing proofs to the more general case in Appendix~\ref{app:Cs-M}.

In the same vein we define the quantity $C_{\sigma,M}=W R_{\sigma,M}$ for any distribution $\sigma$ and subset $M$, and define $C_{S,M}$ analogously.
Similar to the above theorem, we will later prove a connection of this more general quantity to the behavior of a random walk.

\subsection{Quantum walks and quantum walk search algorithms}\label{sec:prelim-qw-search}

Let $D(P)$ denote the discriminant matrix of a random walk on a weighted graph, as defined in \eqref{eq:discriminant}.
We can associate a quantum walk operator to this random walk, which is a unitary operator for which the following holds.

\begin{definition}[Quantum walk operator]\label{def:walk-operator}
For any reversible Markov chain $P$ on a finite state space $X$, a \emph{quantum walk operator} $W(P)$ is a unitary on $\mathcal{H}_A\otimes \mathcal{H}_X$, such that $\ket{\bar{0}}\in \mathcal{H}_A$, $\mathrm{span}\{\ket{u}:u\in X\}\subseteq \mathcal{H}_X$, and for all $u \in X$ it holds that
\[
(\bra{\bar{0}}\otimes \bra{u})W(P)(\ket{\bar{0}}\otimes I_X) = \bra{u}D(P)
\quad \text{ and } \quad
(\bra{\bar{0}}\otimes I_X)W(P)(\ket{\bar{0}}\otimes \ket{u})= D(P)\ket{u}.
\]
\end{definition}
We note that this definition is more general than the usual notion of Szegedy's quantum walk operator \cite{szegedy2004QMarkovChainSearch}.
Nevertheless, this definition perfectly fits Szegedy's framework and its later extensions, and enables obtaining all major results (but with increased generality and clarity). 

In case $\mathcal{H}_X= \mathrm{span}\{\ket{u}:u\in X\}$ we can simply write 
$(\bra{\bar{0}}\otimes I)W(P)(\ket{\bar{0}}\otimes I)= D(P)$
and such a unitary is called a \emph{block-encoding} of $D(P)$. But as indicated by our definition, all of our results also apply if $W(P)$ is a block-encoding of a matrix $M$ that can be block-diagonalised with $D(P)$ being one of its diagonal blocks -- this generalization is relevant for example in  applications where a data-structure is attached to each vertex (for more details see \cite{gilyen2018QSingValTransf}). Thus, for simplicity in the presentation we will use a slight abuse of notation and simply write $(\bra{\bar{0}}\otimes I)W(P)(\ket{\bar{0}}\otimes I)= D(P)$, when $(\bra{\bar{0}}\otimes I)W(P)(\ket{\bar{0}}\otimes I)= M$ and $M$ is block-diagonal, and the block corresponding to the subspace $\mathrm{span}\{\ket{u}:u\in X\}$ is $D(P)$.

To gain some intuition, we recall the usual construction for Szegedy's quantum walk operator \cite{szegedy2004QMarkovChainSearch}.
It starts from the classical walk perspective: a step of a classical random walk from vertex $u$ consists of sampling a new vertex $v$ according to the distribution $P_{u,\cdot}$ given by the $u$-th row of $P$.
An analogous quantum operation is a unitary $V(P)$ on $\mathrm{span}\{\ket{u,v}:u,v\in X\cup\{\bar{0}\}\}$ defined by
\begin{equation} \label{eq:VP}
\ket{\bar{0}}\ket{u}\mapsto V(P) \ket{\bar{0}}\ket{u} = \left( \sum_{v\in X}\sqrt{P_{u,v}}\ket{v} \right) \ket{u},
\end{equation}
and acting arbitrarily (but unitarily, and controlled on the second register) on the rest of the state space.
Using such a unitary, one can simulate the classical random walk by measuring the state and re-initializing the first register to $\bar{0}$  in every step.
Using the unitary swap operator $\textsc{Shift}\colon \ket{u,v}\mapsto\ket{v,u}$ (for $u,v\in X$), we can now define the operator
\begin{equation}\label{eq:VSV}
V(P)^\dagger \, \textsc{Shift} \, V(P),
\end{equation}
where the dagger $\dagger$ denotes the Hermitian conjugate. One can now verify that this operator is indeed a quantum walk operator, as defined in Definition~\ref{def:walk-operator}.

In order to understand the relationship between our perspective and most previous works on Szegedy-type quantum walks, we note that a Szegedy-type quantum walk is usually implemented as the following sequence of gates:
\[
\ldots  \overbrace{V(P)^\dagger \, \textsc{Shift} \, V(P)}^{W(P)}([2\ketbra{\bar{0}}{\bar{0}}-I]\otimes I)V(P)^\dagger \, \textsc{Shift} \, \underbrace{V(P)([2\ketbra{\bar{0}}{\bar{0}}-I]\otimes I)V(P)^\dagger}_{\textsc{Ref}}\ldots
\]
This can be viewed~\cite{szegedy2004QMarkovChainSearch,magniez2006SearchQuantumWalk} as a sequence of unitaries $\ldots \textsc{Shift} \, \textsc{Ref} \, \textsc{Shift} \, \textsc{Ref} \ldots$, where $\textsc{Ref}$ is a reflection around the span of the states in the right-hand side of \eqref{eq:VP}. We instead look at it as the sequence $\ldots W(P) \, ([2\ketbra{\bar{0}}{\bar{0}}-I]\otimes I) \, W(P) \, ([2\ketbra{\bar{0}}{\bar{0}}-I]\otimes I) \ldots$. There are various advantages of our treatment: 
\begin{itemize}
	\item It directly reveals the discriminant matrix, which is at the core of the analysis
	\item It enables the use of new techniques such as fast-forwarding or block-encoding	
	\item There is no need to work with pairs of vertices $\ket{u}\ket{v}$
	\item The similarities and differences compared to the classical walk are more apparent
\end{itemize}

Analogous to the case of classical random walk algorithms, quantum walk search algorithms are assumed to have access to the following (possibly controlled) black-box operations (and their inverses):
\begin{itemize}
\item $\checkv(M)$: checks whether a vertex $u$ is marked. Complexity $\check(M)$ or $\check$.
Described by the mapping
\[
\forall u \in X, b \in \{0,1\}: \quad
\ket{u}\ket{b}
\mapsto
\begin{cases}
\ket{u}\ket{b} &\text{if }u\notin M \\
\ket{u}\ket{b\oplus 1} &\text{if }u\in M.
\end{cases}
\]
\item $\setupv(\pi)$: generates the superposition $\ket{\sqrt{\pi}}=\sum_{u\in X}\sqrt{\pi_u}\ket{u}$.
Complexity $\setup(\pi)$ or $\setup$.
\item $\updatev(P)$: implements a (controlled) walk operator $W(P)$, as in \cref{def:walk-operator}.
Complexity $\update(P)$ or $\update$.
\end{itemize}

\begin{remark}
In the literature the update cost is often defined as the cost implementing the unitary $V(P)$ described in~\eqref{eq:VP}, which is compatible with our cost notion due to \eqref{eq:VSV}. However, in some papers the update cost is defined as the cost of implementing $\textsc{Ref}$, which is harder to compare directly. Still it seems unlikely that $\textsc{Ref}$ can be implemented much more efficiently than $W(P)$, so we do not devote much attention to this minor conflict in the definitions. When we cite such a paper we re-express their bounds in terms of our cost functions, c.f., Belovs' original paper~\cite{belovs2013Electric} on electrical network based quantum walks.
\end{remark}

\paragraph{Implementing interpolated quantum walks.}
We can derive other operations by combining the above black-box operations.
Say that we wish to implement a quantum walk corresponding to the interpolated walk $P(s)$.
Let $\theta = \arccos(\sqrt{s})/2$, and 
\[
V = \begin{pmatrix}
\cos(\theta) & \sin(\theta)\\
\sin(\theta) & -\cos(\theta)
\end{pmatrix}, \quad 
X = \begin{pmatrix}
0 & 1\\
1 & 0
\end{pmatrix}, \quad 
Y = \begin{pmatrix}
0 & 1\\
-1 & 0
\end{pmatrix}.
\]
Then
\[
-YVXV\ket{0}
= \sqrt{s}\ket{0} + \sqrt{1-s}\ket{1},\quad
-YVV\ket{0}
= \ket{1}.
\]
Let $cW(P)$ be a controlled version of the ``update unitary'' $W(P)$ controlled by the first qubit, and $C$ be the ``check unitary'' flipping the first qubit for marked vertices. Then the operator
$$
[I_A\otimes(V\otimes I_{X}) \, C \, (VY\otimes I_{X})] \, cW(P) \, [I_A\otimes (YV\otimes I_{X}) \, C \, (V\otimes I_{X})]
$$
is a quantum walk operator\footnote{In order to show that it satisfies the requirements of Definition~\ref{def:walk-operator}, note that the subspace $\mathrm{span}(\ket{0},\ket{1})\otimes\mathrm{span}\{\ket{u}\colon u\in X \}$ is invariant under the action of $C$.} corresponding to $D(P(s))$.
This shows that we can implement $\updatev(P(s))$ using one call to $\updatev(P)$, two calls to $\checkv(M)$ and 4 elementary gates.
We can hence bound the complexity
\[
\update(P(s)) \in \bigO{\update(P) + \check(M)}.
\]

\subsection{Quantum fast-forwarding}
Similar to interpolated walks, \emph{quantum fast-forwarding}~\cite{apers2018QFastForwardMarkovChains} has proven to be a useful tool for quantum walk search algorithms.
For example, the most recent developments \cite{ambainis2019QW} in the hitting time framework (Theorem~\ref{thm:hitting-time}) are based on this technique.

\begin{theorem}[Quantum fast-forwarding \cite{apers2018QFastForwardMarkovChains}]\label{thm:fast-forwarding}
Let $\eps\in (0,1)$, $t\in \mathbb{N}$ and let $P$ be any reversible Markov chain on state space $X$. There is a quantum algorithm with complexity $\bigO{\sqrt{t\log\frac{1}{\eps}}\mathsf{U}}$, where $\mathsf{U}$ is the cost\footnote{For simplicity we assume that $\mathsf{U}$ is at least the number of qubits on which $W(P)$ act. This is a fair assumption because if $W(P)$ uses fewer gates than qubits, then there must be some unused qubits.} of implementing a quantum walk operator $W(P)$ (and its inverse), that implements a unitary $U$ on $\mathrm{span}\{\ket{a,x}: a \in A,x\in X\}$ for some finite set $A\ni \{\tilde{0}\}$, such that for any $\ket{\psi}\in\mathrm{span}\{\ket{x}:x\in X\}$, 
$$\nrm{(\bra{\tilde{0}}\otimes I)U\ket{\bar{0}}\ket{\psi}- \ket{\tilde{0}}D(P)^t\ket{\psi}}^2\leq \eps.$$
\end{theorem}

The resulting unitary $U$ can be equivalently described as a $(1,\log|A|,\eps)$-block-encoding of $D(P)^t$ \cite{chakraborty2018BlockMatrixPowers,gilyen2018QSingValTransf}.
For completeness we will later reprove this theorem, and give an explicit quantum circuit solving this problem, that will play a crucial role in \cref{sec:alter}.

\section{Quantum walk frameworks} \label{sec:qw-frameworks}
In this section we survey the different quantum walk search frameworks.

\subsection{Hitting time framework} \label{sec:ht-framework}

The hitting time framework is the quantum analogue of arguably the simplest random walk search algorithm:\footnote{In the case of classical algorithms, the subroutines $\mathtt{Setup}$, $\mathtt{Update}$, and $\mathtt{Check}$ are assumed to: sample a vertex according to $\pi$; sample a neighbour of the current vertex $u$; and check if the current vertex is marked. Note that the costs $\csetup$ and $\cupdate$ of the classical operations might be significantly cheaper than their quantum counterparts $\setup$ and $\update$ (since the quantum checking operation is just the reversible version of the classical checking operation, it will never be significantly harder).}
\begin{enumerate}
\item Use $\setupv(\pi)$ to sample a vertex $u$ according to $\pi$.
\item Repeat $\HT$ times:
\begin{enumerate}
\item Check if the current vertex is marked using $\checkv(M)$.
\item Sample a neighbour of the current vertex using $\updatev(P)$, and make that the current vertex.
\end{enumerate}
\end{enumerate}
If $\HT \in \Omega(\HT(P,M))$, then this algorithm finds a marked vertex with constant probability.
Its complexity is $\csetup + \HT(\cupdate + \ccheck)$.

The quantum analogue of this framework was first introduced by Szegedy~\cite{szegedy2004QMarkovChainSearch}, giving a quadratic speedup over the update and checking costs of the classical algorithm.
His algorithm, however, only \emph{detected} the presence of a marked vertex, rather than actually finding one.
Later work by Ambainis et al.~\cite{ambainis2019QW} resolved this issue, describing a quantum walk algorithm that effectively \emph{finds} a marked vertex with constant probability.
We summarize their result in the theorem below.

\begin{theorem}[Hitting time framework]\label{thm:hitting-time}
Let $P$ be any reversible Markov chain on a finite state space $X$, $M\subset X$ a marked set, and $\HT$ a known upper bound on $\HT(P,M)$.
Then there is a quantum algorithm that outputs a vertex $x$ from $M$ with constant probability in complexity
\[
\bigO{\setup \sqrt{\log(\HT)} + \sqrt{\HT}(\update + \check) \sqrt{\log(\HT)\log\log(\HT)}}.
\]
\end{theorem}

\subsection{MNRS framework} \label{sec:mnrs-framework}

While the hitting time framework is optimal in terms of the number of quantum walk steps, it may be suboptimal in the number of calls to $\mathtt{Check}(M)$.
In contrast, the MNRS framework uses an optimal number of calls to $\mathtt{Check}(M)$, at the expense of possibly more calls to $\mathtt{Update}(P)$. This framework is again best understood using a random walk perspective.
Specifically, consider a classical algorithm that, rather than checking after every step, only checks after every $\frac{1}{\delta}$ steps:
\begin{enumerate}
\item Use $\mathtt{Setup}(\pi)$ to sample a vertex $u$ according to $\pi$.
\item Repeat $\frac{1}{\eps}$ times:
\begin{enumerate}
\item Check if the current vertex is marked using $\mathtt{Check}(M)$.
\item Repeat $\frac{1}{\delta}$ times:
\begin{enumerate}
\item[] Sample a neighbour of the current vertex using $\mathtt{Update}(P)$, and make that the current vertex. 
\end{enumerate}
\end{enumerate}
\end{enumerate}
If $\delta$ is a lower bound on the spectral gap of $P$, then by taking $\frac{1}{\delta}$ steps between checks, each random walk vertex that is checked is approximately distributed as an independent sample from $\pi$.
Hence, if $\eps \leq \pi(M)$, then after $\frac{1}{\eps}$ samples from $\pi$ a marked vertex will be sampled with constant probability.
The total complexity of this classical algorithm is $\csetup + \frac{1}{\eps}(\frac{1}{\delta}\cupdate + \ccheck)$.
If $\eps = \pi(M)$ and $\delta$ equals the spectral gap of $P$, then it holds that
\begin{equation} \label{eq:ht-bnds}
\frac{1}{\eps}
\leq
\HT(P,M)
\leq
\frac{1}{\eps\delta}.
\end{equation}
Hence this approach is often suboptimal in the number of steps with respect to the algorithm in the previous section\footnote{E.g., consider the graph consisting of two cliques on $n/2$ nodes, and a single edge between them. For a single marked element we get that $\pi(m) \in \Theta(1/n)$, $\delta(P) \in \Theta(1/n^2)$ and $\HT \in \Theta(1/n^2)$, so that $1/(\epsilon\delta) \gg \HT$.}, but it may be better in the number of checks.
This algorithm will hence be preferable in cases where the checking cost $\ccheck$ is much larger than the update cost $\cupdate$.

The quantum analogue of this classical algorithm was introduced by Magniez, Nayak, Roland and Santha~\cite{magniez2006SearchQuantumWalk}, who showed the following:
\begin{theorem}[MNRS framework]
Let $P$ be any reversible Markov chain on a finite state space $X$, $M \subset X$ a marked set, $\eps$ a known lower bound on $\pi(M)$ and $\delta$ a known lower bound on the spectral gap of $P$.
Then there is a quantum algorithm that outputs a vertex $x$ from $M$ with constant probability in complexity 
\[
\bigO{\setup + \frac{1}{\sqrt{\eps}}\left(\frac{1}{\sqrt{\delta}}\update + \check \right)}.
\]
\end{theorem}

\subsection{Controlled quantum amplification framework}\label{sec:prelim-controlled}
Dohotaru and H{\o}yer \cite{dohotaru2017controlledQAmp} showed that an \emph{optimal} payoff between update cost and checking cost can be obtained for the special case where there is a single marked element $M = \{m\}$.
Specifically, consider the following classical algorithm:
\begin{enumerate}
\item Use $\mathtt{Setup}(\pi)$ to sample a vertex $u$ according to $\pi$.
\item Repeat $1/\eps$ times:
\begin{enumerate}
\item Check if the current vertex is marked using $\mathtt{Check}(M)$.
\item Repeat $\eps \HT$ times:
\begin{enumerate}
\item[] Sample a neighbour of the current vertex using $\mathtt{Update}(P)$, and make that the current vertex. 
\end{enumerate}
\end{enumerate}
\end{enumerate}
If $\eps \in \Theta(\pi(m))$, $\HT \in \Omega(\HT(P,\{m\}))$, and there is a unique marked element $M = \{m\}$, then the above algorithm returns $m$ with constant probability.
This can be proven using the following lemma.
\begin{lemma}[{\cite[Section 6]{dohotaru2017controlledQAmp}}] \label{lemma:DH}
Consider a reversible Markov chain $P$ with stationary distribution $\pi$, and a single marked element $m$.
If $\tau \in \Omega(\pi(m) \HT(P,m))$, then
\[
\HT(P^\tau,m)
\in \bigO{\frac{1}{\pi(m)}}.
\]
\end{lemma}
If we set $\tau = \eps \HT$ then $\HT(P^\tau,m)$ exactly describes the expected number of repetitions of step 2.~in the above algorithm that are required to find a marked element.
The algorithm has complexity of the order $\csetup + {\HT} \cupdate + \frac{1}{{\eps}} \ccheck$.
This corresponds to the update complexity of the classical algorithm in the hitting time framework, and the checking complexity of the classical algorithm in the MNRS framework.

In their \emph{controlled quantum amplification framework}, Dohotaru and H{\o}yer \cite{dohotaru2017controlledQAmp} construct a quantum analogue of the above result.
\begin{theorem}[Controlled quantum amplification framework]
Let $P$ be any reversible Markov chain on a finite state space $X$, $m \in X$ a unique marked element, $\HT$ a known upper bound on $\HT(P,\{m\})$, and $\eps$ a known lower bound on $\pi(m)$.
Then there is a quantum algorithm that outputs $m$ with constant probability in complexity 
\[
\bigOt{\setup + \sqrt{\HT} \, \update + \frac{1}{\sqrt{\eps}} \, \check}.
\]
\end{theorem}

\subsection{Electric network framework}\label{sec:electric-framework}
A main drawback of all the aforementioned frameworks is that they require the quantum walk to start from a quantum sample of the stationary distribution, which may in general be much more difficult to construct than, say, a distribution that is supported on a single vertex.
Belovs \cite{belovs2013Electric} showed that one can combine quantum walks with tools from electric network theory to get rid of this restriction, proving the theorem below.
For an arbitrary distribution $\sigma$, we let $\setup(\sigma)$ denote the complexity of generating the initial state $\ket{\sqrt{\sigma}} = \sum_u\sqrt{\sigma_u}\ket{u}$, which can be much smaller than the setup cost $\setup=\setup(\pi)$ for generating the quantum sample $\ket{\sqrt{\pi}}$ of the stationary distribution.

We also define $\Lambda(\sigma,C)$, for some choice of $C$, as a unitary that acts as:
\begin{equation}\label{eq:Lambda}
\ket{{0}}\ket{u}
\mapsto \frac{\sqrt{\pi_u}\ket{{0}} + \sqrt{\sigma_u/C}\ket{1}}{\sqrt{\pi_u + \sigma_u/C}}\ket{u},
\end{equation}
and let $\mathsf{R}(\sigma)$ be its implementation cost (again including controlled/inverse versions). 
We define $\update(\sigma)=\mathsf{U}+\mathsf{R}(\sigma)$. With these costs, we can describe Belovs' framework as follows.
\begin{theorem}[Electric network framework \cite{belovs2013Electric}] \label{thm:belovs}
Let $P$ be any reversible Markov chain on a finite state space $X$, $M\subset X$ a marked set, $\sigma$ a distribution on $X$, and $C$ a known upper bound on $C_{\sigma,M}$.
Then there is a quantum algorithm that decides if $M\neq \emptyset$ with bounded error in complexity
\[
\bigO{\setup(\sigma) + \sqrt{C}\left(\update(\sigma) + \check\right)}.
\]
\end{theorem}

Recall that if $\sigma = \pi$ then $C = HT(P,M)$.
In addition, $\mathsf{R}(\sigma)$ becomes trivial in this case, so that this recovers the hitting time framework (except that it only \emph{detects} marked elements).

\section{Finding in the electric network framework} \label{sec:electric}
The major drawback of the electric network framework is that it only allows for detecting marked vertices, rather than actually finding one -- in fact, if we only want to detect a marked vertex, the electric network framework is a strict generalization of the hitting time framework.
In this section, we describe a quantum algorithm that reproduces the electric network framework, but in addition actually \emph{finds} marked elements, making it a strict generalization of the hitting time framework.

Recall that we are given a weighted graph $G$ over $X$ with total edge weight $W$, a distribution $\sigma$ over $X$, and a subset $M \subset X$ of marked elements.
The quantity $R_{\sigma,M}$ denotes the effective resistance from $\sigma$ to $M$, and we define $C_{\sigma,M} = WR_{\sigma,M}$.
We will prove the following.
\begin{theorem}[Electric network framework]\label{thm:electric-finding}
Let $P$ be any reversible Markov chain on a finite state space $X$, $M\subset X$ a marked set, $\sigma$ a distribution on $X$, and $C$ a known upper bound on $C_{\sigma,M}$.
There is a quantum algorithm that finds a marked element from $M$, or decides that it is empty, with bounded error in complexity
\[
\bigO{\sqrt{\log C}\,\mathsf{S}(\sigma) + \sqrt{C\log C\log\log C} \left(\mathsf{U}(\sigma)+\mathsf{C}\right)}.
\]
\end{theorem}

In earlier work, Belovs \cite{belovs2013Electric} showed that it is possible to detect the presence of marked vertices in $\bigO{\sqrt{C}}$ quantum walk steps (Theorem~\ref{thm:belovs}).
Our work strengthens this result by also finding a marked element, at the cost of an additional log factor.
Our algorithm and analysis runs along the same lines as \cite{ambainis2019QW}.
As described in Section~\ref{sec:QW-alg}, the algorithm combines quantum fast-forwarding with a quantum walk derived from an interpolated Markov chain.
The analysis, described in Section~\ref{sec:boxes} reduces the success probability of our quantum algorithm to the probability that a classical random walk starting from $\sigma$ hits $M$, and then returns to the support of $\sigma$, within $\bigO{C_{\sigma,M}}$ steps of the walk, similar to the analysis in \cite{ambainis2019QW}.
In order to lower bound this quantity, we extend the known combinatorial interpretation of the quantity $C_{\sigma,M}$ from the case where $\sigma$ is a singleton (\cref{thm:commuteVertex}) to a more general setting.
This is described in Section~\ref{sec:combinatorial}.

\begin{remark}
A similar result, but restricted to the special case where the graph is a tree, can be found in the quantum algorithm for backtracking by Montanaro \cite{montanaro2015quantum}.
Starting from the root of a binary tree, this algorithm incurs an additional log factor for actually finding a solution, rather than simply detecting one.
The extension however crucially relies on the tree structure of the graph, essentially performing a binary search, and hence seems restricted to this special class of graphs.
\end{remark}

\subsection{Combinatorial interpretation of \texorpdfstring{$C_{S,M}$ and $C_{\sigma,M}$}{C(S,M) and C(sigma,M)} }\label{sec:combinatorial}
In Theorem~\ref{thm:commuteVertex} we mentioned the classic result that for any vertex $s$ and subset $M$, the electric quantity $C_{s,M} = WR_{s,M}$ equals the commute time $\Exp_s(\tau^M_s)$ from $s$ to $M$.
A similar interpretation however seems to be lacking for the more general quantity $C_{\sigma,M} = WR_{\sigma,M}$.
While one could expect that a similar relation should hold, at least for the special case where $\sigma = \pi|_S$ equals the stationary distribution on some subset $S$, we provide a counterexample in Appendix~\ref{app:counterexample}.
There we show that in certain cases $C_{\pi|_S,M} = WR_{\pi|_S,M}$ is not equal to the commute time $\Exp_{\pi|_S}(\tau^M_S)$ from $\pi|_S$ to $M$ and back to $S$.

Nevertheless, we do succeed in proving a one-way bound, showing that a variant of the commute time can indeed be bounded by the electric quantity $C_{\pi|_S,M}$, which will prove sufficient for our purpose.

\begin{claim} \label{claim:commute-time}
Let $S\subseteq X\setminus M$, and $p\in \R, T\in \N$ such that $\frac{2}{T}\leq \pi(S)p \leq 1/C_{S,M}$.
Then with probability at least $p/2$ the random walk started from $\sigma=\pi|_{S}$ first hits $M$, and then returns to $S$, in the first $T$ steps.
\end{claim}

\noindent
The conditions on $\sigma$ might seem a bit strong, we can however adapt any graph, similar to what is implicitly done in Belovs' algorithm \cite{belovs2013Electric}, to ensure that they do hold for $p=\frac12,T=4C_{\sigma,M}$.

The main technical contribution in the proof of Claim~\ref{claim:commute-time} is the following lemma, which we prove in Appendix~\ref{app:Cs-M}.
We let $\tau_M$ denote the hitting time of $M$, which is the random variable representing the number of steps to reach $M$ (i.e., the minimum $i$ such that $Y_i\in M$), and $\tau_S^+$ the \emph{first return time} to $S$ (i.e., the minimum $i>0$ such that $Y_i\in S$).

\begin{restatable}{lemma}{lemmaSM}\label{lem:s-M}
	Let $S,M\subseteq X$ be disjoint sets, then
	\[
	\Pr_{\pi|_S}(\tau_M < \tau_S^+)
	= \frac{1}{C_{S,M} \pi(S)}\left(\geq \frac{1}{C_{\pi|_S,M} \pi(S)}\right).
	\]
\end{restatable}
This generalizes the classic fact that the probability that a reversible Markov chain starting at $s$ visits $t$ before returning to $s$ is $1/(C_{s,t}\pi_s)$.
We can then combine this with the following classic lemma, a proof of which can be found in \cite[Lemma 21.13]{levin2017MarkovChainsMixingTimes}.
\begin{lemma}[Kac's Lemma] \label{lem:kac}
For any irreducible and reversible Markov chain it holds that
\[
\Exp_{\pi|_S}(\tau_S^+)
= \frac{1}{\pi(S)}.
\]
\end{lemma}

\noindent The proof of the main claim then easily follows.
\begin{proof}[Proof of Claim~\ref{claim:commute-time}]
We use a union bound on the events that $\tau_M < \tau_S^+$ and $\tau_S^+ < T$, the union of which implies the claimed statement.
From Lemma~\ref{lem:s-M} we know that $\Pr_\sigma(\tau_M < \tau_S^+) = 1/(C_{S,M} \pi(S)) \geq p$.
A bound on $\Pr_\sigma(\tau_S^+ < T)$ easily follows from Kac's lemma (Lemma~\ref{lem:kac}).
Combined with Markov's inequality this lemma implies that
\[
\Pr_\sigma(\tau_S^+ < T)
\geq
\Pr(\tau_S^+ < 2/(\pi(S)p))
> 1-p/2.
\]
The claim then follows by a union bound:
\[
\Pr_\sigma((\tau_M < \tau_S^+) \wedge (\tau_S^+ < T))
\geq \Pr_\sigma(\tau_M < \tau_S^+) + \Pr_\sigma(\tau_S^+ < T) - 1
> p/2. \qedhere
\]
\end{proof}

\subsection{Quantum walk algorithm} \label{sec:QW-alg}

Our algorithm combines quantum fast-forwarding with interpolated quantum walks.
Similar to \cite{ambainis2019QW}, the analysis of the algorithm then follows from a ``box-stretching'' argument, which builds on our combinatorial Claim~\ref{claim:commute-time}.
To ensure the conditions of this claim, we will consider an interpolated walk on a slightly modified graph, implicitly used in~\cite{belovs2013Electric}.

\subsubsection{Modified graph and Belovs' quantum walk} \label{sec:mod-graph}
The input to the search problem is a weighted graph $G = (X,E,w)$, a subset of marked elements $M \subset X$ and an initial distribution $\sigma$ over $X$.
We assume access to these through black-box operations $\checkv(M)$, $\setupv(\sigma)$, and the operator $\Lambda(\sigma,C)$, as defined in \cref{sec:electric-framework}, with $C$ an upper bound on $C_{\sigma,M}$.

We will consider a slightly modified graph $G' = (X',E',w')$, with total weight $W'$, initial distribution $\sigma'$ and marked elements $M'$.
In \cref{lem:Pprime-cost} we will prove that we can implement a quantum walk on $G'$ using only the original black-box operations mentioned earlier.
The graph $G'$ essentially consists of two copies of the original graph, and is defined by
\[
\begin{gathered}
X'
= \{0,1\} \times X, \\
E'
= \{((0,u),(0,v)),\; (u,v) \in E\} \;\cup\; \{((0,u),(1,u)),\; u \in \supp(\sigma) \}, \\
w'_{(0,u),(0,v)} = w_{u,v},\; \forall (u,v) \in E, \qquad w'_{(0,u),(1,u)} = \sigma_u W/C,\; \forall u \in X.
\end{gathered}
\]
We illustrate this construction in Figure~\ref{fig:modified-graph} below.
The modified initial state $\sigma'$ is defined as $\sigma'_{(1,u)} = \sigma_u$, and zero elsewhere, and the marked elements $M' \subseteq V'$ correspond to the set $\{0\} \times M$.

\begin{figure}[htb]
\centering
\includegraphics[draft=false,width=.55\textwidth]{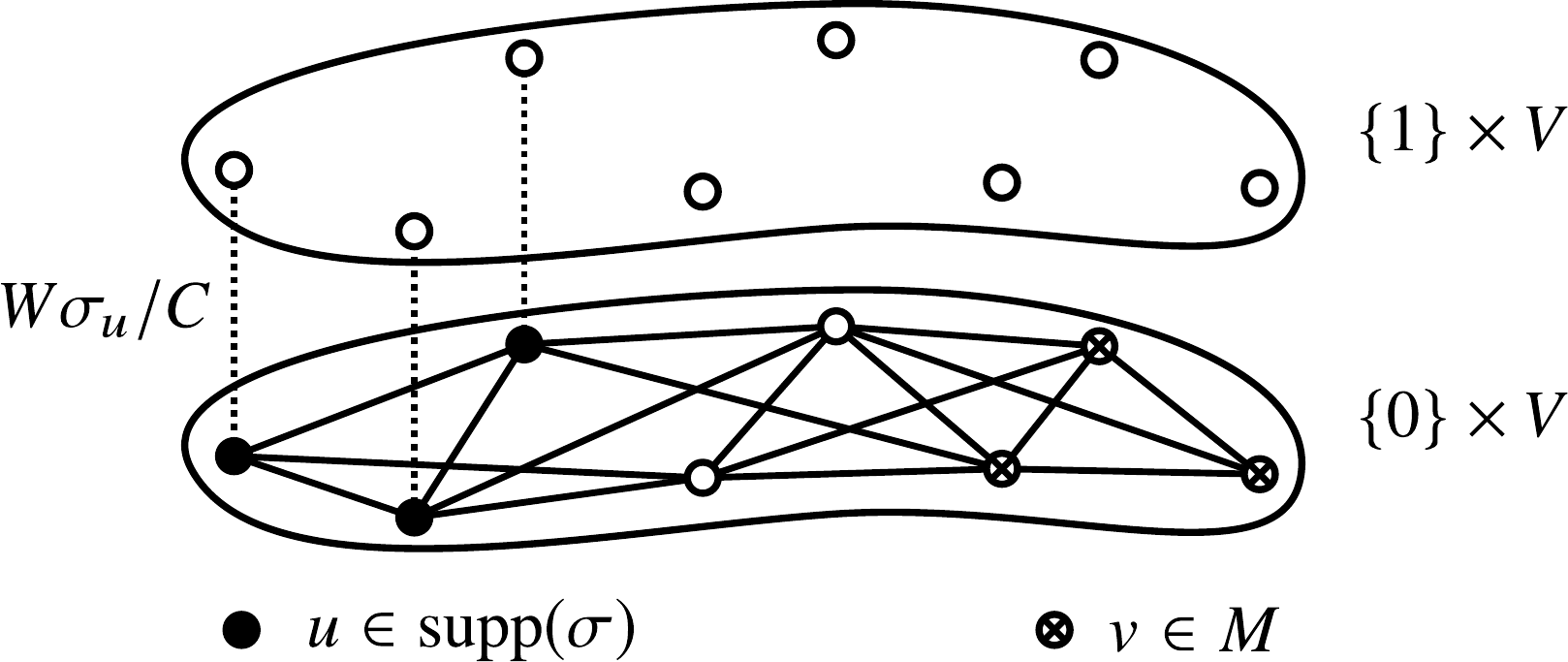}
\caption{Modified graph.}
\label{fig:modified-graph}
\end{figure}

These particular choices ensure that (i) $S':=\mathrm{supp}(\sigma')$ is disjoint from $M'$, (ii) $\sigma'$ is proportional to $\pi'$ on its support, with $\pi'$ the stationary distribution on $G'$, and (iii) the commute time between $\sigma'$ and $M'$ in $G'$ does not increase too much.
We prove the second and third points in the following lemma, letting $R'_{\sigma',M'}$ and $C'_{\sigma',M'}=R'_{\sigma',M'}W'$ denote the effective resistance and commute time, respectively, between $\sigma'$ and $M'$ on $G'$.
We also use the fact that $W' = W + 2W\sum_u \sigma_u/C = W + 2W/C$.

\begin{lemma} \label{lem:props-mod-graph-gen}
	We have $\pi'(S') = 1/(C+2)$, moreover if $\rho$ is a probability distribution on $S=\supp(\sigma)$ such that $\sum_{u\in S}\frac{\rho_u^2}{\sigma_u}= 1/p$, then $C'_{\rho',M'} = (C_{\rho,M}/C + 1/p)(C+2)$.
	Therefore,
	\[
	\frac{1}{C'_{S',M'}\pi'(S') }\geq  \frac{1}{2}\min\left\{p,\frac{C}{C_{\rho,M}}\right\}.
	\]
\end{lemma}
\begin{proof}
	Since $W' = W + 2W/C$, and $w_{(1,u)}'=W\sigma_u/C$, we get that
	\[
		\pi'_{(1,u)}
		= \frac{w_{(1,u)}'}{W'}
		= \frac{W\sigma_u}{W'C}
		= \frac{\sigma_u}{(1+2/C)C}
		= \frac{\sigma_u}{(C+2)}
		\quad \Longrightarrow \quad
		\pi'(S')=\frac1{(C+2)}.
	\]
	Now observe that 
	\[
	R'_{\rho',M'}
	= R_{\rho,M} + \sum_{u}\frac{\rho_u^2}{W\sigma_u/C}
	=  R_{\rho,M} + \frac C{Wp},
	\]
	yielding
	\[
	C'_{\rho',M'}
	= W' R'_{\rho',M'}
	= (W + 2W/C) (R_{\rho,M} + C/(Wp)) 
	=(C+2) (C_{\rho,M}/C + 1/p).
	\]
	Finally,
	\[
	\frac{1}{C'_{S',M'}\pi'(S') }\geq\frac{1}{C'_{\rho',M'}\pi'(S') }= \frac{1}{C_{\rho,M}/C + 1/p}\geq \frac{1}{2\max\{C_{\rho,M}/C, 1/p\}}=\frac{1}{2}\min\{p,C/C_{\rho,M}\}.	 \qedhere
	\]
\end{proof}

Now we apply the above lemma with $\rho=\sigma$, which gives $p=1$. If $C \in \Theta(C_{\sigma,M})$ and $C_{\sigma,M}=\Omega(1)$, then we see that $C'_{\sigma',M'} \in \Theta(C_{\sigma,M})$, so the commute time does not increase significantly.
Moreover, if $C\geq C_{\sigma,M}$, then for $T:=\lceil 4 C'_{\sigma',M'} \rceil\leq 8(C+2)+1$ we have
\[
\frac{2}{T}
\leq \pi'(S')/2
\leq \frac{1}{C'_{\sigma',M'}}.
\]
In this case, all conditions of Claim~\ref{claim:commute-time} are satisfied with respect to $p'=\frac{1}{2}$.

Finally, we show that we can efficiently implement a quantum walk corresponding to the random walk $P'$ on $G'$:

\begin{lemma}\label{lem:Pprime-cost}
	The cost of $\updatev(P')$ is $\update(P') \in \bigO{\update(P) + \mathsf{R}(\sigma)}=\bigO{\update(\sigma)}$, where $\update(P)$ is the cost of implementing the original quantum walk operator $W(P)$, and $\mathsf{R}(\sigma)$ is the cost of $\Lambda(\sigma,C)$, defined in Section~\ref{sec:prelim-electric}. 
\end{lemma}
\begin{proof}
    In the following we will identify $P$ with $\ketbra{0}{0} \otimes P$ (and similarly for $D(P)$).
 	Observe that 
 	$$
 	P'= \diag{\frac{\pi_u}{\pi_u+\sigma_u/C}}P
 	+ \sum_{u \in \supp(\sigma)}\frac{\sigma_u/C}{\pi_u+\sigma_u/C}\left(\ketbra{0,u}{1,u}+\ketbra{1,u}{0,u}\right),
 	$$
 	and therefore $D(P')=\sqrt{P'\circ P'^T}$ expressed in terms of  $D(P)=\sqrt{P\circ P^T}$ is
 	\begin{align*}
 	D(P')=&\, \diag{\sqrt{\frac{\pi_u}{\pi_u+\sigma_u/C}}}D(P)\diag{\sqrt{\frac{\pi_u}{\pi_u+\sigma_u/C}}} \\
 	& + \sum_{u}\sqrt{\frac{\sigma_u/C}{\pi_u+\sigma_u/C}}\left(\ketbra{0,u}{1,u}+\ketbra{1,u}{0,u}\right).
 	\end{align*}
	As in \eqref{eq:Lambda} let $\Lambda(\sigma,C)$ be a unitary that uses a single-qubit ``flag'' register $a$ and acts as 
	$$\ket{\bar{0}}_{A'}\ket{u}_{X}\mapsto \frac{\sqrt{\pi_u }\ket{\bar{0}}_{a}+\sqrt{\sigma_u/C}\ket{1}_{a}}{\sqrt{\pi_u + \sigma_u/C}}\ket{u}_X.$$
 	We will use a new qubit register $b$, and represent the vertices as $(0,u)\mapsto \ket{0}_b\ket{u}$, $(1,u)\mapsto \ket{1}_b\ket{u}$. We will denote by $\bar{c}_{ab}W(P)$ the operator $W(P)$ conditioned on the qubit state $\ket{0}_a\ket{0}_b$, and by $\bar{c}_b\Lambda$ the operator $\Lambda(\sigma,C)$ conditioned on the qubit state $\ket{0}_b$, and using $a$ as the output qubit. Let 
	$$
	W(P'):=
	(I_A\otimes(\bar{c}_b\Lambda^\dagger))\bar{c}_{ab}W(P)(I_A\otimes (\mathrm{SWAP}_{ab}\otimes I_X)\bar{c}_b\Lambda),
	$$
	then
	\begin{align*}
	(\bra{0}_a\otimes I)W(P')(\ket{0}_a\otimes I)
	=&\,
	\ketbra{0}{0}_b\otimes \left[\diag{\sqrt{\frac{\pi_u}{\pi_u+\sigma_u/C}}}W(P)\diag{\sqrt{\frac{\pi_u}{\pi_u+\sigma_u/C}}}\right]
	\\
	&+\left(\ketbra{0}{1}_b+\ketbra{1}{0}_b\right)\otimes I_A\otimes \left[\sum_{u}\sqrt{\frac{\sigma_u/C}{\pi_u+\sigma_u/C}}\ketbra{u}{u}\right],
	\end{align*}
	and therefore $W(P')$ is a walk operator of $D(P')$ whenever $W(P)$ is a walk operator of $D(P)$.
\end{proof}

\subsubsection{Interpolated walk and algorithm}

The transformation described in Section~\ref{sec:mod-graph} can be applied to any input graph $G$ and any distribution $\sigma$, with negligible impact on $C_{\sigma,M}$, or the update cost $\mathsf{U}$.
This justifies focusing our analysis on the case where $\sigma = \pi|_S$ and $1/C_{\sigma,M}\leq \pi(S)\leq 2/C_{\sigma,M}$, ensuring that the conditions of Claim~\ref{claim:commute-time} are satisfied for $p=1/2$ and $T\geq 4C_{\sigma,M}$.
Moreover, this ensures that $M$ and $S=\supp(\sigma)$ are disjoint, and furthermore, that we can easily reflect around $\mathrm{supp}(\sigma)$.\footnote{To see the second point, note that in the modified graph of Section~\ref{sec:mod-graph} this amounts to reflecting around the states whose first register is 1.}
We assume these conditions for the remainder of this section.

For a Markov chain $P$, and parameter $q=(q_S,q_M) \in [0,1)^2$, we consider the interpolated Markov chain $P(q)$, defined by (here $\delta_{uv}$ is the Kronecker delta)
\[
P(q)_{u,v}
= \begin{cases}
(1-q_S) P_{u,v} + q_S \delta_{uv} & \mbox{if } u \in \supp(\sigma) \\
(1-q_M)P_{u,v} + q_M \delta_{uv} &\mbox{if }u\in M\\
P_{u,v} & \mbox{else.}
\end{cases}
\]
Equivalently,
\begin{equation} \label{eq:Pq}
P(q)
= (1-q_S-q_M)P + q_S P_{\supp(\sigma)} + q_M P_M,
\end{equation}
with $P_{\supp(\sigma)}$ and $P_M$ absorbing walks, as defined in Section~\ref{sec:prelim-interpolated}.
We denote by $D(q)$ the discriminant matrix of $P(q)$.
Starting from the state $\ket{\sqrt{\sigma}}$, our algorithm will apply $T$ steps of quantum fast-forwarding of $P(q)$ for appropriately chosen $q$, and increasing~$T$.

\begin{algorithm}[H]
$\mathbf{Search}(P,\sigma,M,T)$.
\begin{enumerate}
\item
Let $Q = \big\{ 1,2^{-1},2^{-2},\dots,2^{-\lceil\log(14T)\rceil} \big\}$, and prepare the state
\[
\ket{\psi}
= \sum_{t\in[T]}\sum_{q_M \in Q} \frac{1}{\sqrt{T|Q|}} \ket{t}\ket{q = (1-T/2,1-q_M)} \ket{\sqrt{\sigma}}.
\]
\item \label{step:ffw}
Let $U$ be the operator that applies quantum fast-forwarding, controlled on the first two registers, mapping $\ket{t}\ket{q}\ket{\sqrt{\sigma}}$ to $\ket{1}\ket{t}\ket{q}D^t(q)\ket{\sqrt{\sigma}}+\ket{0}\ket{\Gamma}$ for some arbitrary $\ket{\Gamma}$, with precision $\bigO{\frac{1}{\log T}}$.
\item
Apply $\bigO{\sqrt{\log T}}$ rounds of amplitude amplification on $U$ and $\ket{\psi}$, conditioned on the first register.
Finally, measure the last register.
\end{enumerate}
\caption{Fast-forwarding-based search algorithm}\label{alg:commute}
\end{algorithm}

\noindent
In Lemma~\ref{lem:stop} of Section~\ref{sec:boxes} we will prove that there exists some $T' \in \bigO{C_{\sigma,M}}$, such that for all $T''\geq T'$ the algorithm returns a marked vertex with constant probability.
Combined with the following lemma, which bounds the complexity of the algorithm, this proves \cref{thm:electric-finding}.
\begin{lemma} \label{lem:compl}
The complexity of \cref{alg:commute} is
\[
\bigO{\sqrt{\log T} \, \setup(\sigma)+\sqrt{T \log T \log\log T} \, (\mathsf{U}(\sigma)+\mathsf{C})}.
\]
\end{lemma}
\begin{proof}
For the complexity of step 1., note that creating $\ket{\psi}$ only requires $\bigO{\log T}$ elementary gates, and a call to $\setupv(\sigma)$ costing $\setup(\sigma)$. For the complexity of step 2., note that by Theorem~\ref{thm:fast-forwarding}, the operators $U$ and $U^\dag$ require $\bigO{\sqrt{T \log\log T}}$ calls to $\updatev(P(q))$, which implements a block-encoding of $W(q)=W(P(q))$. We can implement such a block-encoding, using a block-encoding of $P$, $W(P)$, which, by assumption (see also Lemma~\ref{lem:Pprime-cost}), can be implemented in $\bigO{\mathsf{U}(\sigma)}$ complexity; the operation $\mathtt{Check}(M)$, costing $\check$; and an analogous operation that checks if a vertex is in $S=\mathrm{supp}(\sigma)$, which can be done in $\bigO{1}$ cost, by our assumptions on the structure of $G$ (such an implementation of $W(q)$ is straightforward, as we discuss in Section~\ref{sec:prelim-qw-search}).
Thus, the total cost of step 2. is $\bigO{\sqrt{T\log\log T}(\update(\sigma)+\check)}$.

By \cite{brassard2002AmpAndEst} we can implement step $3.$ using $\bigO{\sqrt{\log T}}$ reflections around $U\ket{\psi}$.
A single such reflection can be implemented by $\bigO{1}$ calls to $U$ and $U^\dag$, and the preparation circuit of $\ket{\psi}$ and its inverse -- yielding a total complexity (neglecting constants):
\begin{equation*}
\sqrt{\log T}\left( \setup(\sigma)+\sqrt{T\log\log T}(\update(\sigma)+\check) \right) = \sqrt{\log T}\setup(\sigma)+\sqrt{T\log T\log\log T}(\update(\sigma)+\check). \qedhere
\end{equation*}
\end{proof}

\subsubsection{Correctness of Algorithm~\ref{alg:commute}} \label{sec:boxes}
Similar to the argument in \cite{ambainis2019QW}, we use a careful choice of the parameters of the interpolated walk to ensure a constant success probability of Algorithm~\ref{alg:commute}.
As discussed in the previous section, we can assume without loss of generality that $\sigma = \pi|_S$ for some $S \subseteq X$, and that $1/C_{\sigma,M}\leq \pi(S) \leq 2/C_{\sigma,M}$.

The key quantity is $\nrm{\Pi_M D^t(q) \ket{\sqrt{\sigma}}}^2$, where $\Pi_M$ is the orthogonal projector onto marked vertices -- this is the probability of finding a marked vertex in step $2.$ of the algorithm (that is, before amplitude amplification).
We can bound this quantity in terms of the classical Markov chain $P(q)$ as follows, generalizing \cite[Lemma 8]{ambainis2019QW}:
\begin{lemma} \label{lem:nrm-bnd}
Let $(Y_i{(q)})_{i=0}^\infty$ be a Markov chain evolving according to $P(q)$ with initial state $Y_0(q)$ distributed according to $\sigma=\pi|_S$, and let $D(q)$ be the associated discriminant matrix.
Then for any $t,t' \in \N$ such that $t'>t$ we have that
\[
\nrm{\Pi_M D^t(q) \ket{\sqrt{\sigma}}}
\geq \Pr_{Y_0(q)\sim \sigma}\big( Y_t(q) \in M, Y_{t'}(q) \in S\big).
\]
\end{lemma}
\noindent
We will not use this lemma directly, but we state it here for the sake of intuition. Its proof closely follows that of \cite[Lemma 8]{ambainis2019QW} (see also the proof of Corollary~\ref{cor:bnd-exp}).

By this lemma it suffices to show that with certain probability, for appropriate choice of $t$ and $t'$, the $t$-th vertex is marked and the $t'$-th vertex is again in the initial support $S$.
We will be able to ensure these conditions by appropriately tweaking the parameters $q=(q_S,q_M)$.

The sequence of random variables $(Y_i)_{i=0}^{\infty}$ is supported on infinite sequences of vertices from $X$, which represent \emph{paths} of the random walk. However, in light of Lemma~\ref{lem:nrm-bnd}, given such a sequence, we will only care about which states in the sequence are in $S$, and which are in $M$. Following a similar abstraction in \cite{ambainis2019QW}, we model a path of the random walk as a sequence of boxes, with gray boxes representing vertices in $S$, black boxes representing vertices in $M$, and white boxes representing vertices in neither $S$ nor $M$. We depict such a sequence of boxes in Figure~\ref{fig:sequence}. In a slight abuse of notation, we will refer to $y=(y_0,y_1,\dots)$ drawn from $(Y_i)_{i=0}^{\infty}$ as a sequence of boxes. 
A gray or black box at position $i$ then denotes the event that $Y_i \in S$ or $Y_i \in M$, respectively.
The indicated times $\hit$ and $\ct$ in Figure~\ref{fig:sequence} denote the random variables corresponding to the hitting time (first time to reach $M$) and commute time (first time to return to $S$ after reaching $M$), respectively. 

\begin{figure}[H]
\centering
\begin{tikzpicture}[scale=.3]
\filldraw[fill=gray] (0,0) rectangle (1,1);
\draw (1,0) rectangle (2,1);
\draw (2,0) rectangle (3,1);
\draw (2,0) rectangle (4,1);
\draw (4,0) rectangle (5,1);
\draw (5,0) rectangle (6,1);
\draw (5,0) rectangle (7,1);
\filldraw[fill=gray] (7,0) rectangle (8,1);
\draw (8,0) rectangle (9,1);
\draw (9,0) rectangle (10,1);
\draw (10,0) rectangle (11,1);
\draw (11,0) rectangle (12,1);
\draw (12,0) rectangle (13,1);
\draw (13,0) rectangle (14,1);
\draw (14,0) rectangle (15,1);
\draw (15,0) rectangle (16,1);
\draw (16,0) rectangle (17,1);
\draw (17,0) rectangle (18,1);
\filldraw (18,0) rectangle (19,1);
\node at (18.5,2) {$\hit$};
\draw (19,0) rectangle (20,1);
\draw (20,0) rectangle (21,1);
\draw (21,0) rectangle (22,1);
\filldraw (22,0) rectangle (23,1);
\filldraw (23,0) rectangle (24,1);
\draw (24,0) rectangle (25,1);
\draw (25,0) rectangle (26,1);
\draw (26,0) rectangle (27,1);
\draw (27,0) rectangle (28,1);
\draw (28,0) rectangle (29,1);
\filldraw[fill=gray] (29,0) rectangle (30,1);
\node at (29.5,2) {$\ct$};
\draw (30,0) rectangle (31,1);
\draw (31,0) rectangle (32,1);
\draw (32,0) rectangle (33,1);
\draw (33,0) rectangle (34,1);
\filldraw[fill=gray] (34,0) rectangle (35,1);
\draw (35,0) rectangle (36,1);
\draw (35,0) rectangle (37,1);
\draw (37,0) rectangle (38,1);
\draw (38,0) rectangle (39,1);
\draw (39,0) rectangle (40,1);
\draw (40,0) rectangle (41,1);
\draw (41,0) rectangle (42,1);
\draw (42,0) rectangle (43,1);
\draw (43,0) rectangle (44,1);
\filldraw (44,0) rectangle (45,1);
\draw (45,0) rectangle (46,1);
\draw (46,0) rectangle (47,1);
\filldraw[white] (46.5,-.1) rectangle (47.1,1.1);
\node at (47.5,.5) {$\dots$};
\end{tikzpicture}
\caption{Sequence of boxes $y=(y_0,y_1,\dots)$ drawn from the random variable $(Y_i)_{i=0}^{\infty}$.
Gray boxes, called $S$-boxes, correspond to $y_k \in S$; black boxes, called $M$-boxes, correspond to $y_k \in M$.
The hitting time from $S$ to $M$ is denoted by $\hit$, the commute time by $\ct$.}\label{fig:sequence}
\end{figure}
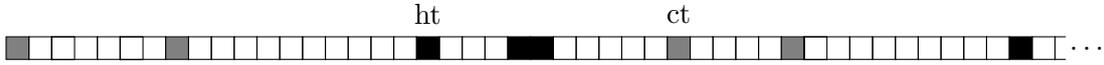

In the following we define $r_S = 1/(1-q_S)$ and $r_M = 1/(1-q_M)$, representing the expected number of steps that the interpolated walk remains at a vertex in $S$ or $M$ respectively.
Given parameters $r=(r_S,r_M)$ and a sequence of boxes $y = (y_0,y_1,\dots)$, we denote by $\gamma^{(r)}$ the sequence derived from $y$ by replacing each $M$-box with $r_M$ $M$-boxes, and each $S$-box with $r_S$ $S$-boxes. Since $r_M$ is the expected number of steps a walker would stay at a marked vertex in an absorbing walk with parameter $q_M$ (and similarly for $r_S$), $\gamma^{(r)}$ loosely models a path of the interpolated random walk $P(q)$. 
We denote by $\hit^{(r_S,r_M)}$ and $\ct^{(r_S,r_M)}$  the hitting time and commute time of the sequence~$\gamma^{(r)}$.

\begin{figure}[H]
\centering
\begin{tikzpicture}[scale=.3]
\filldraw[fill=gray] (0,0) rectangle (1,1);
\filldraw[fill=gray] (1,0) rectangle (2,1);
\draw (2,0) rectangle (3,1);
\draw (2,0) rectangle (4,1);
\draw (4,0) rectangle (5,1);
\draw (5,0) rectangle (6,1);
\draw (5,0) rectangle (7,1);
\draw (7,0) rectangle (8,1);
\filldraw[fill=gray] (8,0) rectangle (9,1);
\filldraw[fill=gray] (9,0) rectangle (10,1);
\draw (10,0) rectangle (11,1);
\draw (11,0) rectangle (12,1);
\draw (12,0) rectangle (13,1);
\draw (13,0) rectangle (14,1);
\draw (14,0) rectangle (15,1);
\draw (15,0) rectangle (16,1);
\draw (16,0) rectangle (17,1);
\draw (17,0) rectangle (18,1);
\draw (18,0) rectangle (19,1);
\draw (19,0) rectangle (20,1);
\filldraw (20,0) rectangle (21,1);
\node at (21.5,2.25) {$\hit^{(r_S,r_M)}$};
\draw[->] (20.5,1.75)--(20.5,1);
\filldraw (21,0) rectangle (22,1);
\filldraw (22,0) rectangle (23,1);
\draw (23,0) rectangle (24,1);
\draw (24,0) rectangle (25,1);
\draw (25,0) rectangle (26,1);
\filldraw (26,0) rectangle (27,1);
\filldraw (27,0) rectangle (28,1);
\filldraw (28,0) rectangle (29,1);
\filldraw (29,0) rectangle (30,1);
\filldraw (30,0) rectangle (31,1);
\filldraw (31,0) rectangle (32,1);
\draw (32,0) rectangle (33,1);
\draw (33,0) rectangle (34,1);
\draw (34,0) rectangle (35,1);
\draw (35,0) rectangle (36,1);
\draw (35,0) rectangle (37,1);
\filldraw[fill=gray] (37,0) rectangle (38,1);
\node at (38.7,2.25) {$\ct^{(r_S,r_M)}$};
\draw[->] (37.5,1.75)--(37.5,1);
\filldraw[fill=gray] (38,0) rectangle (39,1);
\draw (39,0) rectangle (40,1);
\draw (40,0) rectangle (41,1);
\draw (41,0) rectangle (42,1);
\draw (42,0) rectangle (43,1);
\filldraw[fill=gray] (43,0) rectangle (44,1);
\filldraw[fill=gray] (44,0) rectangle (45,1);
\draw (45,0) rectangle (46,1);
\draw (46,0) rectangle (47,1);
\filldraw[white] (46.5,-.1) rectangle (47.1,1.1);
\node at (47.5,.5) {$\dots$};
\end{tikzpicture}
\caption{If $y$ is the sequence of boxes shown in Figure~\ref{fig:sequence}, then the above sequence represents $\gamma^{(r_S,r_M)}$ for $r_S=2$ and $r_M=3$.}\label{fig:sequence2}
\end{figure}
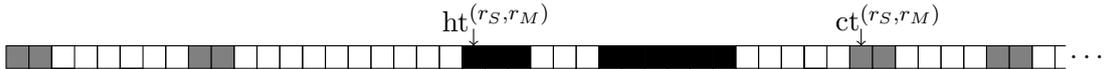

For integers $a<b$, we will be interested in bounding the quantities
\[
M_y^{(r)}[a,b]
= |\{t \in [a,b]: \gamma^{(r)}_t \in M\}|
\quad \text{ and } \quad
S_y^{(r)}[a,b]
= |\{t \in [a,b]: \gamma^{(r)}_t \in S\}|.
\]
When $r=(1,1)$ we will omit the $(r)$ superscript.

\begin{lemma} \label{lem:comb}
Let $y = (y_0,y_1,\dots,y_T,\dots)$ denote a sequence of boxes, for some $T\in \N$, such that
\begin{itemize}
\item $y_0\in S$,
\item
$\ct \leq T$,
\item
$S_y[0,\hit] = 1$.
\end{itemize}
If we set $r_S = T/2$, then there exists $r_M \in R = \{1,2,4,\dots,2^{\lceil \log(14T)\rceil}\}$ such that
\[
M_y^{(r)}[0,2T] \geq T/2
\quad \text{and} \quad
S_y^{(r)}[7T,15T] \geq T/4.
\]
\end{lemma}
\noindent We note that $R=\{1/q_M:q_M\in Q\}$ for $Q$ the set defined in Algorithm~\ref{alg:commute}.
\begin{proof}
We first choose $r_M \in \{1,2,4,\dots,2^{\lceil \log(14T)\rceil}\}$ such that $7T \leq \ct^{(1,r_M)} \leq 14T$.
To see that this is possible, note that $\ct^{(1,1)}=\ct\leq 14T$ by assumption, and increasing $r_M$ strictly increases $\ct^{(1,r_M)}$, so there is some largest possible $r_M$ such that $\ct^{(1,r_M)}\leq 14T$.
Since doubling $r_M$ can at most double $\ct^{(1,r_M)}$, we must also have $7T\leq \ct^{(1,r_M)}$. 

The increase in the commute time when transforming $y$ to $\gamma^{(1,r_M)}$, given by $\ct^{(1,r_M)}-\ct$, comes from adding $\ct^{(1,r_M)}-\ct$ new $M$-boxes to $\gamma^{(1,r_M)}$ before $\ct^{(1,r_M)}$, meaning that 
\begin{align*}
M_y^{(1,r_M)}[0,\ct^{(1,r_M)}] &\geq \ct^{(1,r_M)}-\ct\geq \ct^{(1,r_M)}-T.
\end{align*}
Note that for any positive integer $k$, we have $M_y^{(1,r_M)}[0,\ct^{(1,r_M)}-k]\geq M_y^{(1,r_M)}[0,\ct^{(1,r_M)}]-k$.
Setting $k=\ct^{(1,r_M)}-3T/2$ thus implies that
\begin{align*}
M_y^{(1,r_M)}[0,3T/2] &\geq T/2.
\end{align*}

Next we choose $r_S = T/2$.
By the conditions $S_y[0,\hit]=1$ and $y_0\in S$, the unique $S$-box in $y$ before $\hit$ is $y_0$, so the second $S$-box in $y$ is at $\ct$, and similarly, the second $S$-box in $\gamma^{(1,r_M)}$ is at $\ct^{(1,r_M)}$. Thus, there are $r_S-1=T/2-1$ new boxes added to $\gamma^{(1,r_M)}$ before $\ct^{(1,r_M)}$ to get $\gamma^{(r_S,r_M)}$, meaning that the $\geq T/2$ $M$-boxes in the first $3T/2$ boxes of $\gamma^{(1,r_M)}$ must all occur within the first $2T$ boxes of $\gamma^{(r_S,r_M)}$, so 
\[
M_y^{(r_S,r_M)}[0,2T]
\geq T/2.
\]
Similarly, since the first $S$-box in $\gamma^{(1,r_M)}$ is extended to $T/2$ $S$-boxes to get to $\gamma^{(r_S,r_M)}$, the second $S$-box is displaced by $T/2-1$, meaning $\ct^{(r_S,r_M)}=\ct^{(1,r_M)}+T/2-1$, so 
\begin{equation*}
15T/2 - 1 \leq \ct^{(r_S,r_M)}\leq 29T/2-1.
\end{equation*}
In $\gamma^{(r_S,r_M)}$, there is a sequence of $r_S=T/2$ $S$-boxes beginning at $\ct^{(r_S,r_M)}$, so we have:
\[
S_y^{(r_S,r_M)}[7T,15T]
\geq T/2. \qedhere
\]
\end{proof}

For a fixed Markov process $P$, and parameters $r=(r_S,r_M)$, let $(Y_i^{(r)})_{i=0}^\infty$ be the Markov chain evolving according to $P(q)$, the absorbing chain with $q_S=1-\frac{1}{r_S}$ and $q_M=1-\frac{1}{r_M}$ (i.e., $Y^{(r)}=Y(q)$).

For a sequence $y$, for any choice of $r=(r_S,r_M)$, we associate two slightly different sequences of boxes to $y$:
\begin{itemize}
\item
As previously defined, $\gamma^{(r)}$ is derived from $y$ by replacing each $M$- resp.~$S$-box with a fixed number of copies $r_M$ resp.~$r_S$.
\item
$y^{(r)}$ is derived from $y$ by replacing each $M$- resp.~$S$-box with an independently random number of copies that are geometrically distributed with mean $r_M$ resp.~$r_S$.
Note that if $y$ is distributed according to the random variable $(Y_i)_{i=0}^{\infty}$, then $y^{(r)}$ is distributed according to $(Y_i^{(r)})_{i=0}^{\infty}$. 
\end{itemize}

Lemma~\ref{lem:comb} allows us to say something about the probability that $\gamma^{(r)}_t\in M$ and $\gamma^{(r)}_{t'}\in S$ under uniform random $r\in R$, $t\in [0,2T]$, and $t'\in [7T,15T]$, assuming certain conditions on $y$ are satisfied (we will shortly argue that these conditions are satisfied with high probability when starting from the distribution $\sigma=\pi|_S$). Since we are actually concerned with proving statements about the absorbing walk, we would like to say something similar for $y^{(r)}$ in place of $\gamma^{(r)}$. Fortunately, these two sequences are very similar, leading to the following corollary of Lemma~\ref{lem:comb}. 

\begin{corollary} \label{cor:bnd}
Set $r_S = T/2$, and let $r_M \in R$, $t\in \{1,\dots,30T\}$ and $t'\in \{1,\dots,30T\}$ be chosen uniformly at random. 
Let $Y=(Y_i)_{i=0}^{\infty}$ be the Markov chain evolving according to Markov process $P$ (not necessarily reversible) starting in any distribution $\kappa$ supported on $S$, with $Y^{(r)}$ the absorbing Markov chain, coupled to $Y$, as defined above.
Let $E$ be the event that: $\ct \leq T$ and $S_Y[0,\hit] = 1$, where $\hit$ and $\ct$ are the hitting time and commute time in $Y$.
Then
\[
\mathbb{E}_{t,t',r_M}\left(\Pr_{Y_0^{(r)} \sim \kappa}(Y_t^{(r)}\in M, Y_{t+t'}^{(r)} \in S \mid E)\right)
= \Omega(\log^{-1} T).
\]
\end{corollary}
\begin{proof}
By Lemma~\ref{lem:comb}, we know that if $E$ holds then there exists some $r_M \in R$, chosen with probability $1/|R| = \Omega(\log^{-1} T)$, such that with probability 1 we have 
\[
|\{i \in [0,2T]: \gamma^{(r)}_i \in M\}|
= M_Y^{(r)}[0,2T] \geq T/2
\]
and
\[
|\{i \in [7T,15T]: \gamma^{(r)}_i \in S\}|
= S_Y^{(r)}[7T,15T] \geq T/4.
\]
We will show that for \emph{any} $y$ in the support of $Y$ such that $E$ holds, given this choice of $r$, $\Pr_{t,t'}(y_t^{(r)}\in M,y_{t+t'}^{(r)}\in S)$ is constant, completing the proof. 
Using the similarities between $y^{(r)}$ and $\gamma^{(r)}$, we now wish to lower bound $|\{i\in[0,7T/2]:y^{(r)}_i\in M\}|$ and $|\{i\in[7T/2+1,30T]:y_i^{(r)}\in S\}|$.

In going from $y$ to $\gamma^{(r)}$, we replace each $S$-box with a block of $r_S$ $S$-boxes, whereas when going from $y$ to $y^{(r)}$, we replace each $S$-box with a block of $S$-boxes whose length is geometrically distributed with mean $r_S$ (and similarly for $M$-boxes). We can equivalently derive $y^{(r)}$ from $\gamma^{(r)}$ by replacing each block of $r_S$ $S$-boxes with a block of $S$-boxes whose length is given by a geometric sample. 
Define $\overline{r}_S[a,b]$ to be the mean over all such geometric samples replacing an $S$-block of $\gamma^{(r)}$ that has non-trivial overlap with the interval $[a,b]$. Define $\overline{r}_M[a,b]$ similarly. 

We will assume the following 6 conditions:
\begin{align} \label{eq:cond-geoms}
\begin{split}
\overline{r}_M[1,2T],\; \overline{r}_M[1,7T],\; \overline{r}_M[7T,15T] &\in [r_M/2,2r_M],\\
\overline{r}_S[1,2T],\; \overline{r}_S[1,7T],\; \overline{r}_S[7T,15T] &\in [r_S/2,2r_S].
\end{split}
\end{align}
We use \cite[Lemma 10]{ambainis2019QW}, which states that any sum of i.i.d.~geometric random variables will be within a factor 2 of its expected value with probability at least $7/16$.
While for instance the sums $\overline{r}_M[1,2T]$ and $\overline{r}_M[1,7T]$ are not strictly independent, clearly the event that $\overline{r}_M[1,2T] \in [r_M/2,2r_M]$ cannot decrease the probability that $\overline{r}_M[1,7T] \in [r_M/2,2r_M]$.
Hence the probability that all 6 conditions hold is at least $(7/16)^6$, and we will assume that this is the case for the rest of the proof.

Since $r_S/2\leq \overline{r}_S[0,2T]\leq 2r_S$ and $r_M/2\leq \overline{r}_M[0,2T]\leq 2r_M$, it holds that 
\begin{equation}
|\{i\in[0,7T/2]:y_i^{(r)}\in M\}|
\geq M_y^{(r)}[0,2T]/2
\geq T/4.
\label{eq:M-bound}
\end{equation}
To see this, notice that in the interval $[0,2T]$  of $\gamma^{(r)}$, we remove at most half of the $\geq T/2$ $M$-boxes to get $y^{(r)}$, so there are at least $T/4$ remaining. 
We at most double the $\leq 3T/2$ $S$-boxes, so the remaining $\geq T/4$ $M$-boxes from $\gamma^{(r)}$'s interval $[0,2T]$ all appear within the first $2T+3T/2=7T/2$ of $y^{(r)}$. From \eqref{eq:M-bound}, we immediately have
\begin{equation*}
|\{i\in[1,30T]:y_i^{(r)}\in M\}|
\geq T/4,
\end{equation*}
so the probability that a uniform random $t\in\{1,\dots,30T\}$ satisfies $y^{(r)}_t\in M$ is at least $1/120$.

Next, since $\overline{r}_M[0,7T],\overline{r}_M[7T,15T] \in [r_M/2,2r_M]$ and $\overline{r}_S[0,7T],\overline{r}_S[7T,15T] \in [r_S/2,2r_S]$, 
\begin{equation}
|\{i\in[7T/2,30T]:y_i^{(r)}\in S\}|
\geq |\{i\in[7T,15T]:\gamma_i^{(r)}\in S\}|/2
= S_y^{(r)}[7T,15T]/2
\geq T/8.
\label{eq:S-bound}
\end{equation}
To see this, note that in the interval $[7T,15T]$ of $\gamma^{(r)}$, we remove at most half of the $\geq T/4$ $S$-boxes to get $y^{(r)}$, so there are at least $T/8$ remaining, although they might no longer be contained within the interval $[7T,15T]$. However, since the position of any element in $y^{(r)}$ is at least half and at most double its position in $\gamma^{(r)}$, these $T/8$ $S$-boxes will be within the interval $[7T/2,30T]$ of $y^{(r)}$. 

From \eqref{eq:S-bound}, we want to conclude that $\Pr_{t,t'}(y^{(r)}_{t+t'}\in S|y^{(r)}_t\in M)$ is constant.
In fact, by \eqref{eq:M-bound}, we have that with constant probability $t\leq 7T/2$ \emph{and} $y^{(r)}_t\in M$.
Hence it is sufficient to lower bound $\Pr_{t,t'}(y^{(r)}_{t+t'}\in S|t\leq 7T/2, y^{(r)}_t\in M)$ by a constant.  
To do so, we note that for \emph{any} $t\in [1,\dots,7T/2]$ the range of possible values of $t+t'$ contains $[7T/2,30T]$.
Hence, by \eqref{eq:S-bound}, 
\begin{equation*}
\Pr_{t,t'}(y^{(r)}_{t+t'}\in S|t\leq 7T/2, Y^{(r)}_t\in M)\geq (T/8)/(30T)={1}/{240}.
\end{equation*}
This bound only holds conditioned on the events in \eqref{eq:cond-geoms}, but we already argued that these also hold with constant probability.
\end{proof}

This corollary allows us to prove the following statement.
\begin{corollary} \label{cor:bnd-exp}
Set $r_S=T/2$, and let $r_M\in R$ and $t\in \{1,\dots,30T\}$ be chosen uniformly at random. Let $q_S=1-\frac{1}{r_S}$ and $q_M=1-\frac{1}{r_M}$, and let $D(q)$ be the discriminant of $P(q)$ for a reversible ergodic Markov process $P$ on $X$ with stationary distribution $\pi$, and $\sigma=\pi|_S$ for some $S\subset X$ with $1/C_{\sigma,M}\leq \pi(S)\leq 2/C_{\sigma,M}$. 
If $T \geq 4C_{\sigma,M}$ then 
\[
\mathbb{E}_{t,r_M} \big[ \nrm{\Pi_M D^t(q) \ket{\sqrt{\sigma}}}^2 \big]
\in \Omega(\log^{-1} T).
\]
\end{corollary}
\begin{proof}
Let $Y^{(r)}=Y(q)$ be the Markov chain of the absorbing walk $P(q)$ starting from the distribution $\sigma$. Then we can define $Y$ as the Markov chain evolving according to $P$, coupled to $Y^{(r)}$ as above. That is, $Y$ follows the same sequence as $Y^{(r)}$, except that it omits repeated elements that result from using the absorbing self-edges. 

Let $E$ denote the event that $\ct \leq T$, where $\ct$ is the commute time of $Y$, and $S_Y[0,\hit] = 1$.
Then by Claim~\ref{claim:commute-time} we know that $E$ holds with probability at least $1/4$.
Combined with Corollary~\ref{cor:bnd}, taking $t'$ uniformly at random from $\{1,\dots,30T\}$, this allows us to conclude
\begin{equation}
\mathbb{E}_{t,t',r}\left(\Pr_{Y_0^{(r)} \sim \sigma}(Y_t^{(r)}\in M, Y_{t'}^{(r)} \in S)\right)
\geq \frac{1}{4} \Omega(\log^{-1} T). \label{eq:exp-lb}
\end{equation}
Next, we compute:
\begin{align}
\mathbb{E}_{t,r_M} \big[ \nrm{\Pi_M D^t(q) \ket{\sqrt{\sigma}}}^2 \big]
&= \frac{1}{|R|} \sum_{r_M \in R}\frac{1}{30T}\sum_{t=1}^{30T} \nrm{\Pi_M D^t(q) \ket{\sqrt{\sigma}}}^2 \nonumber\\
&\geq \frac{1}{|R|} \sum_{r_M\in R} \left( \frac{1}{30T}\sum_{t=1}^{30T} \nrm{\Pi_M D^t(q) \ket{\sqrt{\sigma}}} \right)^2 \nonumber\\
&= \frac{1}{|R|}\sum_{r_M\in R}\frac{1}{30T}\sum_{t=1}^{30T}\nrm{\Pi_M D^t(q) \ket{\sqrt{\sigma}}} \frac{1}{30T}\sum_{t'=1}^{30T}\nrm{\Pi_M D^{t'}(q) \ket{\sqrt{\sigma}}}\nonumber\\
&\geq \frac{1}{|R|} \sum_{r_M\in R} \frac{1}{30T}\sum_{t=1}^{30T}\frac{1}{30T}\sum_{t'=1}^{30T}\left|\bra{\sqrt{\sigma}}D^{t}(q)\Pi_M D^{t'}(q)\ket{\sqrt{\sigma}}\right|,\label{eq:exp-nrm}
\end{align}
by the Cauchy-Schwarz inequality.
Let $\pi'$ denote the stationary distribution of $P(q)$. Since $P(q)$ is just a twice interpolated walk, its stationary distribution has the form 
$$\pi'=a\pi|_S+b\pi|_M+c\pi|_{X\setminus (M\cup S)}=a\sigma+b\pi|_M+c\pi|_{X\setminus (M\cup S)},$$
for some positive constants $a$, $b$ and $c$, whose precise values are not important for our purposes (see \cite{krovi2010QWalkFindMarkedAnyGraph} for an analysis of the stationary distribution of interpolated walks). Furthermore, we observe that $D(q)^t=\mathrm{diag}(\pi')^{1/2}P(q)^t\mathrm{diag}(\pi')^{-1/2}$, so, since $\ket{\sqrt{\sigma}}$ is supported on $S$:
\begin{align*}
\bra{\sqrt{\sigma}} D(q)^{t}\Pi_M &= \bra{\sqrt{\sigma}}\mathrm{diag}(a\sigma)^{1/2}P(q)^{t} \mathrm{diag}(b\pi|_M)^{-1/2}\\
&= \sqrt{a/b} \sum_{u\in S}\sigma_u\bra{u} P(q)^{t} \mathrm{diag}({ \pi|_M})^{-1/2}.
\end{align*}
Similarly, 
\begin{equation*}
\Pi_M D(q)^{t'}\ket{\sqrt{\sigma}} = \sqrt{b/a}\mathrm{diag}(\pi|_M)^{1/2} P(q)^{t'} \sum_{u\in S}\ket{u}.
\end{equation*}
Thus, 
\begin{align*}
\left|\bra{\sqrt{\sigma}}D^{t}(q)\Pi_M D^{t'}(q)\ket{\sqrt{\sigma}}\right| 
&= \sum_{u\in S}\sigma_u\bra{u}P(q)^{t}\mathrm{diag}(\pi|_M)^{-1/2}\mathrm{diag}(\pi|_M)^{1/2}P(q)^{t'}\sum_{u\in S}\ket{u}\\
&= \sum_{u\in S}\sigma_u\bra{u}P(q)^{t}\Pi_M P(q)^{t'}\sum_{u\in S}\ket{u}\\
&= \Pr_{Y_0(q)\sim \sigma}(Y_t(q)\in M, Y_{t'+t}(q)\in S).
\end{align*}
Recall that $Y^{(r)}=Y(q)$. Thus continuing, from \eqref{eq:exp-nrm}, and using \eqref{eq:exp-lb}, we have:
\begin{align*}
\mathbb{E}_{t,r_M} \big[ \nrm{\Pi_M D^t(q) \ket{\sqrt{\sigma}}}^2 \big]
&= \frac{1}{|R|} \sum_{r_M\in R} \frac{1}{30T}\sum_{t=1}^{30T}\frac{1}{30T}\sum_{t'=1}^{30T} \Pr_{Y_0(q) \sim \sigma}\big( Y_t(q) \in M, Y_{t'+t}(q) \in S \big) \\
&= \mathbb{E}_{t,t',r_M}\left(\Pr_{Y_0(q) \sim \pi_S}\big( Y_t(q) \in M, Y_{t+t'}(q) \in S\big)\right)
\in \Omega(\log^{-1} T). \qedhere
\end{align*}
\end{proof}

\noindent
From this corollary, we can straightforwardly prove our final lemma.
Combined with Lemma~\ref{lem:compl} this proves \cref{thm:electric-finding}.
\begin{lemma} \label{lem:stop}
There exists $T' \in \bigO{C_{\sigma,M}}$ such that, for all $T \geq T'$, Algorithm~\ref{alg:commute} returns a marked element with constant probability.
\end{lemma}
\begin{proof}
By the above Corollary~\ref{cor:bnd-exp} we know that for all $T \geq T'$, for some $T' \in \bigO{C_{\sigma,M}}$, it holds that
\[
\frac{1}{T}\sum_{t\in[T]} \frac{1}{|Q|}\sum_{q_M\in Q} \nrm{\Pi_M D^t(q) \ket{\sqrt{\sigma}}}^2 = \frac{1}{T}\sum_{t\in [T]}\frac{1}{|R|}\sum_{r_M\in R}\nrm{\Pi_M D^t(q)\ket{\sqrt{\sigma}}}^2
\in \Omega(\log^{-1}T).
\]
As a consequence, measuring the state
\[
\ket{1} \Big(\sum_{t\in[T]} \sum_{q_M\in Q} \frac{1}{\sqrt{T|Q|}} \ket{t}\ket{q}D^t(q)\ket{\pi_S}\Big)
+ \ket{0}\ket{\Gamma}
\]
returns a marked element with probability $\Omega(\log^{-1} T)$.
In step \ref{step:ffw}.~of the algorithm we approximate this state up to sufficient precision $\bigO{\log^{-1} T}$.
Applying $\bigO{\sqrt{\log T}}$ rounds of amplitude amplification then indeed suffices to retrieve a marked element with constant probability.
\end{proof}

\subsection{A generalization of Belovs' algorithm}
Now we sketch a generalization of \cref{alg:commute}. 

\begin{corollary}
	Let $\sigma, \rho$ be probability distributions on $S=\supp(\sigma)\subseteq X$, and $p:=1/\sum_{u\in S}\frac{\rho_u^2}{\sigma_u}$, then there is a quantum algorithm that finds a marked element from $M$ in expected complexity
	\[
	\bigO{\sqrt{1/p}\,\left[\log(C_{\rho,M})\mathsf{S}(\sigma) + \sqrt{C_{\rho,M}} \left(\mathsf{U}(\sigma)+\mathsf{C}\right)\right]\mathrm{polylog}(C_{\rho,M}/p)}.
	\]
\end{corollary}
\begin{proof}
	Suppose that $p C_{\rho,M}\leq C \in \bigO{p C_{\rho,M}}$, then by \cref{lem:props-mod-graph-gen} and \cref{claim:commute-time} we have that the walk started from $\sigma'$ on the modified graph first hits $M'$ and then returns to $S'$ with probability at least $p/4$ within the first $T=\lceil \frac{4}{p}(C+2)\rceil=\bigO{C_{\rho,M}}$ steps. The analysis in \cref{sec:boxes} shows that \cref{alg:commute} finds a marked element with probability $\Omega(p)$ if we enhance the precision by a factor of $\sim p$ in step \ref{step:ffw}. After applying $\sqrt{1/p}$ additional rounds of amplitude amplification we find a marked element with probability $\Omega(1)$. 
	
	Finally note that we do not need to a priori know $\rho$ or the values of $p$ and $C_{\rho,M}$. We can do binary search to find multiplicative constant approximations of $p$ and $C_{\rho,M}$, only incurring a logarithmic overhead and providing an expected runtime as claimed.
\end{proof}
\anote{Another generalization could be to sample from an on-demand distribution on $M$. A similar algorithm should work on an analogously modified graph, with new vertices $(2,m)\colon m\in M$, and with some runtime analogous to $\HT^+$.}

Intuitively this improvement is somewhat analogous to the $\HT^+$ to $\HT$ improvement in the complexity of finding marked elements using quantum walks~\cite{ambainis2019QW}. There the $\HT^+$ complexity corresponds to ``fair'' sampling of a marked vertex from $\pi|_M$, whereas here the runtime $\sqrt{C_{\sigma,M}}$ corresponds to the ``democratic'' requirement that $M$ should be reachable from the entire $\sigma$ -- but one does not actually need to hit the marked set from everywhere! It is enough if we hit it with high probability from a large fraction of the initial states, cf.~\cref{lem:s-M}. Indeed, if $C=C_{\sigma|_Q,M}$ then for $\rho:=\sigma|_Q$ we get $p=\sigma(Q)$, and so we get an efficient algorithm with runtime $\sim \sqrt{C}$ as long as $p$ is not too small, while the quantity $\sqrt{C_{\sigma,M}}$ is less relevant.
\anote{Would be nice to express the condition in a more intuitive way, like $\Pr_\sigma[\tau_M^S\leq C]\geq p$ or something similar. Any ideas?}

To illustrate that this result can be helpful in some cases, we consider the following example.	Suppose that $G$ is a regular graph, with marked set $M$ and hitting time $\HT$. Suppose that we can remove an edge of $G$ without affecting the hitting time much. Take, say $3$ copies of $G$, and cyclically connect to each other the vertices adjacent to the removed edges, so that the graphs form a triangle, with a single edge between each pair of the copies of $G$.  Suppose that we unmark the marked vertices of one copy, and set the weight of the three new edges very small, so that the hitting time in the new graph can be arbitrary large. If the vertices are also permuted there is no apparent structure left, and previous quantum walk algorithms seem to fail in finding or even detecting marked vertices faster than $\widetilde{\mathcal{O}}(\sqrt{\HT'})$, where $\HT'$ is the hitting time of the new graph. However, our algorithm can find a marked vertex in time $\widetilde{\mathcal{O}}(\sqrt{\HT})$. Thus our walk actually finds a marked vertex much faster than the hitting time $\widetilde{\mathcal{O}}(\sqrt{\HT'})$ of the new graph.
\anote{Other applications? Backtracking? Space-time trade-offs in Element distinctness?}

\section{The MNRS framework and the electric network framework}\label{sec:unified}

In this section we describe our second main result, which is a quantum walk search algorithm that generalizes the MNRS framework, the hitting time framework, the controlled quantum amplification framework, and (our extension of) the electric network framework.
It is summarized in the following theorem.

\begin{theorem}\label{thm:mnrs-electric}
For any reversible Markov chain $P$ on state space $X$, any marked set $M\subset X$, any $t\in \N$, and any distribution $\sigma$ on $X$, there is a quantum algorithm that finds a marked element with bounded error in complexity
\begin{align*}
&\bigO{\sqrt{\log(C(t))}\mathsf{S}(\sigma)+\sqrt{C(t)\log(C(t))\log\log(C(t))}(\sqrt{t}\mathsf{U}\sqrt{\log(C(t))}+ \refl(\sigma)+\mathsf{C})}\\
={}& \bigOt{\mathsf{S}(\sigma)+\sqrt{C(t)}(\sqrt{t}\mathsf{U}(\sigma)+\mathsf{C})},
\end{align*}
where $C(t)$ is a known upper bound on $C_{\sigma,M}(P^t)$, $\mathsf{S}(\sigma)$ is the cost of $\mathtt{Setup}(\sigma)$, $\mathsf{U}$ is the cost of the walk operator $W(P)$, $\refl(\sigma)$ is as in Theorem~\ref{thm:belovs}, $\mathsf{U}(\sigma)=\mathsf{U}+\refl(\sigma)$,
and $\mathsf{C}$ is the cost of the $\mathtt{Check}(M)$ operation.
\end{theorem}

Using standard techniques, we can also handle the case where $C_{\sigma,M}(P^t)$ is unknown, at the cost of an additional $\log(C_{\sigma,M}(P^t))$ factor on the first term, giving, for any $t$, an algorithm with complexity (neglecting log factors):
$$\mathsf{S}(\sigma)+\sqrt{C_{\sigma,M}(P^t)}(\sqrt{t}\mathsf{U}(\sigma)+\mathsf{C}).$$
Setting $t=1$, we recover the electric network framework.
In the special case where $\sigma$ equals the stationary distribution $\pi$ of $P$, and thus also of $P^t$, we have $C_{\pi,M}(P^t)=\HT(P^t,M)$ and $\refl(\sigma) \in \bigO{1}$, and so the complexity of the algorithm is (neglecting log factors):
$$\mathsf{S}+\sqrt{\HT(P^t,M)}(\sqrt{t}\mathsf{U}+\mathsf{C}).$$
Setting $t=1$ recovers the hitting time framework (Theorem~\ref{thm:hitting-time}), and setting $t=1/\delta$ recovers the MNRS framework. To see this, note that since $\frac{1}{\delta}$ is at least the mixing time of $P$, a single step of $P^{1/\delta}$ approximately samples from $\pi$, which finds a marked vertex with probability $\eps$, so $\HT(P^{1/\delta},M) = \bigO{1/\eps}$.
If in addition there is a unique marked element $M = \{m\}$, we can choose $t \in \Omega(\eps \HT(P,\{m\}))$ to retrieve the controlled quantum amplification framework.
This immediately follows from Lemma~\ref{lemma:DH} which shows that $\HT(P^t,\{m\}) \in \bigO{1/\eps}$ if $t \in \Omega(\eps \HT(P,\{m\}))$.
If we could extend this bound to larger sets, then we find an immediate and strict extension of their framework.

\begin{proof}[Proof of Theorem~\ref{thm:mnrs-electric}]
We will apply Theorem~\ref{thm:electric-finding} to the reversible Markov chain $P^t$. This gives an algorithm for finding an element $x\in M$ with complexity:
\begin{align}
{}&\mathsf{S}(\sigma)\sqrt{\log(C(t))}+\sqrt{C(t)\log(C(t))\log\log(C(t))}(\mathsf{U}_t+\refl_{\sigma}+\mathsf{C}),\label{eq:mnrs-HT}
\end{align}
where $\mathsf{U}_t$ is the complexity of implementing the walk operator $W(P^t)$, and $\refl_{\sigma}$ is as described above Theorem~\ref{thm:belovs}.
We need only describe how to implement a walk operator $W(P^t)$, and upper bound its complexity $\update_t$.

By Theorem~\ref{thm:fast-forwarding}, since $D(P^t)=D(P)^t$, there is an $\eps$-approximate walk operator $W(P^t)$ for $P^t$ with complexity $\bigO{\sqrt{t\log(1/\eps)}\update}$.
We will call this operator a number of times
\[
\tau
= \sqrt{C(t)\log(C(t))\log\log(C(t))}.
\]
Hence if we set $\eps=\Theta(\frac{1}{\tau})$ then this ensures that the algorithm is correct with bounded error.
This gives 
\[
\update_t
= \bigO{\sqrt{t} \update \sqrt{\log\tau}}
= \bigO{\sqrt{t} \update \sqrt{\log(C(t))}}.
\]
Plugging this into \eqref{eq:mnrs-HT} completes the proof.
\end{proof}

\section{Alternative algorithm for finding in the hitting time framework}\label{sec:alter}

Our quantum walk algorithm relies on the use of quantum fast-forwarding.
This makes it more complicated than the original quantum walk algorithms in e.g.~\cite{szegedy2004QMarkovChainSearch,magniez2006SearchQuantumWalk,belovs2013Electric}.
In this section we show that the correctness of our algorithm implies the correctness of a much simpler algorithm, at least in the regimes corresponding to the hitting time framework and the electric network framework.
Namely, if we simply pick a random interpolation parameter, run the corresponding quantum walk for about $\sqrt{C}$ steps, and finally measure the walk register, then we find a marked element with constant probability.
This algorithm was proposed in \cite[Section~4]{ambainis2019QW} for the hitting time framework, but was only conjectured to be correct.

To derive this result, we literally ``dissect'' the more complicated fast-forwarding algorithm (\cref{alg:commute}) by considering an explicit construction of the quantum fast-forwarding routine.
The structural properties of this quantum circuit then imply that the simpler routine should also succeed. To illustrate this, we give a proof of the correctness of the fast-forwarding technique, \cref{thm:fast-forwarding}, as this is the main tool used in Algorithm~\ref{alg:commute}.

\begin{proof}[Proof of fast-forwarding scheme, \cref{thm:fast-forwarding}]
In order to describe our construction we recall some well-known properties of quantum walks. One of the important basic observations is that for any unitary $W$ for which 
$D:=(\bra{\bar{0}}\otimes I)W(\ket{\bar{0}}\otimes I)$ is a Hermitian matrix it holds~\cite{childs2008OnRelContDiscQuantWalk,gilyen2018QSingValTransf,ambainis2019QW} that 
\begin{equation}\label{eq:ChebyWalk}
(\bra{\bar{0}}\otimes I)\left(([I-2\ketbra{\bar{0}}{\bar{0}}]\otimes I)W^\dagger([I-2\ketbra{\bar{0}}{\bar{0}}]\otimes I)W\right)^{\!n}(\ket{\bar{0}}\otimes I)
= T_{2n}(D),
\end{equation}
where $T_{2n}(x)$ is the $2n$-th Chebyshev polynomial of the first kind.

An intriguing property of Chebyshev polynomials is that~\cite{sachdeva2014FasterAlgsViaApxTheory} 
\begin{equation}\label{eq:binom}
x^t=\sum_{i=0}^{t}2^{-t}\binom{t}{i} T_{2i-t}(x).
\end{equation}
For $t,d\in\mathbb{N}$ even numbers, now define the polynomial 
\[
p_{t,d}(x)
= \sum_{n=-\frac{d}{2}}^{\frac{d}{2}} 2^{-t} \binom{t}{\frac{t}{2}+n} T_{2n}(x),
\]
which is simply the sum in \eqref{eq:binom} truncated.
By Chernoff's bound and \eqref{eq:binom} it follows that for all $\eps>0$, $d\geq \left\lceil \sqrt{2t\ln(2/\eps)} \right\rceil$, and $x\in [-1,1] \colon$ 
\[
|x^t-p_{t,d}(x)| \leq \eps.
\]
Since $T_n(x)=T_{-n}(x)$, by \eqref{eq:ChebyWalk}, for all even $t$ we get that 
\begin{equation}\label{eq:PolyWalk}
p_{t,d}(D)=\sum_{n=-\frac{d}{2}}^{\frac{d}{2}}2^{-t}\binom{t}{\frac{t}{2}+n}(\bra{\bar{0}}\otimes I)\left(([I-2\ketbra{\bar{0}}{\bar{0}}]\otimes I)W^\dagger([I-2\ketbra{\bar{0}}{\bar{0}}]\otimes I)W\right)^{\!|n|}(\ket{\bar{0}}\otimes I).
\end{equation}
Let $\ell\in\mathbb{N}$, we define $C_k:=([I-2(I_{k-1}\otimes\ketbra{1}{1}\otimes I_{\ell-k})\otimes\ketbra{\bar{0}}{\bar{0}}]\otimes I)$ as the controlled reflection operator controlled by the $k$th qubit, where $I_m$ denotes the identity operator on $m$ qubits. Let
\begin{equation}\label{eq:contWalk}
\mathbb{U}^{(\ell)}:=\prod_{k=0}^{\ell-1}\left(C_k W^\dagger C_kW\right)^{\!2^{k}}=\sum_{n=0}^{2^{\ell}-1} \ketbra{n}{n}\otimes\left(([I-2\ketbra{\bar{0}}{\bar{0}}]\otimes I)W^\dagger([I-2\ketbra{\bar{0}}{\bar{0}}]\otimes I)W\right)^{\!n}.
\end{equation}
Now we use the linear combination of unitaries (LCU)~\cite{childs2012HamSimLCU,berry2013ExpPrecHamSim} technique.
Suppose that $d < 2^{\ell+1}$, and $R$ is a unitary such that $R\colon \sqrt{\alpha} \ket{0}\mapsto\sqrt{2^{-t}\binom{t}{t/2}}\ket{0}+\sum_{n=1}^{\frac{d}{2}}\sqrt{2^{1-t}\binom{t}{\frac{t}{2}+n}}\ket{n}$, where $\alpha\in [1-\eps,1]$ is a normalizing factor. A simple LCU calculation shows, that we have
\begin{equation*}
p_{t,d}(D)=\alpha(\bra{0}R^\dagger\otimes\bra{\bar{0}}\otimes I)\mathbb{U}_\ell (R\ket{0}\otimes \ket{\bar{0}}\otimes I), 
\end{equation*}
and therefore setting $\ket{\tilde{0}}:=\ket{0}\otimes \ket{\bar{0}}$ and
\begin{equation}\label{eq:finalFF}
U:=(R^\dagger\otimes I)\mathbb{U}_\ell (R\otimes I)
\end{equation}
we get
\begin{equation*}
\nrm{D^t-\alpha(\bra{\tilde{0}}\otimes I)U(\ket{\tilde{0}}\otimes I)}\leq\eps.
\end{equation*}
We note that the case of odd $t$ can be handled completely analogously using odd counterparts of \eqref{eq:ChebyWalk}-\eqref{eq:contWalk}.
The $\alpha$ factor is also trivial to remove using simple techniques~\cite{gilyen2018QSingValTransf}; alternatively one can apply the triangle inequality and use the slightly weaker error-bound $ \nrm{D^t-(\bra{\tilde{0}}\otimes I)U(\ket{\tilde{0}}\otimes I)}\leq2\eps.$
\end{proof}

Now we are ready to prove the main statement of this section.
We first recall the main technical corollary (\cref{cor:bnd-exp}) underlying \cref{thm:electric-finding}: let $T \geq c C_{\sigma,M}$ for a sufficiently large constant $c$, set $r_S=(T/30)/2=T/60$ and let the other interpolation parameter $r_M \in R = \{1,2,4,\dots,2^{\lceil \log(14T)\rceil}\}$ and time parameter $t\in [T]$ be chosen uniformly at random.
Let $U$ be a block-encoding of $D(q)^t=(\bra{\bar{0}}\otimes I)U(\ket{\bar{0}}\otimes I)$, with $D(q)$ the discriminant matrix of $P(q)$ defined in \eqref{eq:Pq}.
Then measuring the state $U(\ket{\bar{0}}\otimes \ket{\sqrt{\sigma}})$ returns a marked element with probability at least $\Omega\left(\frac{1}{\log(T)}\right)$. This is precisely why Algorithm~\ref{alg:commute} is correct.

In particular we can use the unitary in \eqref{eq:finalFF} when $\eps=\Theta\left(\frac{1}{\log(T)}\right)$ is small enough. Note that since we are only interested in the measurement statistics of the second register we can also use $\mathbb{U}_\ell (R\otimes I)$ instead of $U$. Then measuring the first part of the first register commutes with $\mathbb{U}_\ell$, so we can measure this register already before applying $\mathbb{U}_\ell$, without modifying the measurement statistics. Now we have the following algorithm: Apply $R$ on the first half of the first register, then measure it. Finally apply $\mathbb{U}_\ell$ and measure the second register. But this is again equivalent to first (classically) sampling $n\in [-\frac{d}{2},\frac{d}{2}]$ distributed $\propto 2^{-t}\binom{t}{\frac{t}{2}+n}$, and then applying $2n$ quantum walk steps to the initial state and measuring the second register. This works in the case when $t$ is even; one can also handle odd $t$ analogously by slightly tweaking the circuit $U$. In fact one can show that sampling an even $t\in [T]$ uniformly at random also works in the algorithm of \cite{ambainis2019QW}, which is an alternative solution.

We summarize the resulting algorithm:
\begin{algorithm}[H]
\begin{enumerate}
\item pick $r_M \in R$ and $t\in [T]$ uniformly at random
\item sample $n$ according to $2^{-t}\binom{t}{t/2 + n}$, conditioned on $|n| \in \bigO{\sqrt{T\log(T)}}$ and having the same parity as $t$
\item apply $|n|$ steps of the interpolated quantum walk $W(P(q))$ with $q_M=1-\frac{1}{r_M}$ and $q_S=1-\frac{60}{T}$ to the state $\ket{\sqrt{\sigma}}$
\item measure the second register
\end{enumerate}
\caption{Simple quantum walk algorithm}\label{alg:simpler}
\end{algorithm}

\begin{theorem}
There exists a constant $c$ such that if $T \geq cC_{\sigma,M}$ then Algorithm~\ref{alg:simpler} returns a marked vertex with probability~$\Omega\left(\frac{1}{\log T}\right)$.
\end{theorem}

Repeating this procedure $\Omega(\log T)$ times returns a marked vertex with constant probability.
This yields an algorithm that only uses ordinary (interpolated) quantum walks and finds a marked element with constant probability.
If $T \in \Theta(C_{\sigma,M})$, the algorithm has complexity $\bigO{(\setup(\sigma) + \sqrt{C_{\sigma,M}\log C_{\sigma,M}}(\update(\sigma) + \check))\log(C_{\sigma,M})}$. In the case of $\sigma=\pi$, we are in the hitting time framework, and this complexity becomes
$\bigO{(\setup + \sqrt{\HT\log \HT}(\update + \check))\log(\HT)}$.

\subsection*{Acknowledgments}

We thank Fr\'ed\'eric Magniez, Stephen Piddock and J\'er\'emie Roland for fruitful discussions and useful pointers.

\bibliographystyle{alpha}
\bibliography{Bibliography}

\appendix

\section{Counterexample} \label{app:counterexample}
It is a classic result that the combinatorial commute time $\Exp_s(\tau^t_s)$ equals the electric quantity $C_{s,t} = WR_{s,t}$.
We show that this result can be extended to the case where $t$ is a set $M$, rather than a singleton (see Appendix~\ref{app:Cs-M}).
Similarly one could expect that something of the form $\Exp_{\pi|_S}(\tau^M_S) = WR_{\pi|_S,M}$ should hold.
We show here that in fact this does not hold in general.
Similarly we show that $\Pr_{\pi|_S}(\tau_M < \tau_S^+) \neq \frac{1}{C_{\pi|_S,M} \pi(S)}$, whereas this does hold when $S$ is a singleton.

Let $G$ be a path on three nodes $u - v - w$ with unit weights.
Let $S = \{u,v\}$ and $M = \{w\}$, so that $\pi|_S = \frac{1}{3} e_u + \frac{2}{3} e_v$ and $\pi(S) = 3/4$.
The optimal (and only) $\pi|_S$-$M$ flow pushes value $1/3$ along the edge $(u,v)$, and value $1$ along the edge $(v,w)$.
The effective resistance thus equals $R_{\pi|_S,M} = \frac{1}{3^2} + 1 = \frac{10}{9}$.
Since $W = 4$, this shows that $WR_{\pi|_S,M} = \frac{40}{9}$ and $\frac{1}{C_{\pi|_S,M} \pi(S)} = 3/10$.

On the other hand, we can easily calculate that
\[
\Pr_{\pi|_S}(\tau_M < \tau_S^+)
= \frac{1}{3}
> \frac{1}{C_{\pi|_S,M} \pi(S)},
\]
since the only possibility is to start from $v$ and take the edge $(v,w)$, which happens with probability $\frac{2}{3} \frac{1}{2} = \frac{1}{3}$.
Similarly, we can show that the combinatorial commute time
\[
\Exp_{\pi|_S}(\tau^M_S)
= \frac{39}{9}
< WR_{\pi|_S,M}.
\]
To see this, note that $\Exp_{\pi|_S}(\tau^M_S) = \Exp_{\pi|_S}(\tau_M) + 1$, with $\Exp_{\pi|_S}(\tau_M)$ the expected hitting time of $M$ (after the walk hits $M$, it necessarily jumps back to $S$).
On its turn, $\Exp_{\pi|_S}(\tau_M) = \frac{1}{3} \Exp_u(\tau_M) + \frac{2}{3} \Exp_v(\tau_M)$ and $\Exp_u(\tau_M) = 1 + \Exp_v(\tau_M)$ (a walk from $u$ necessarily jumps to $v$ after one step).
To calculate $\Exp_v(\tau_M)$, note that $\Exp_v(\tau_M) = \frac{1}{2} + \frac{1}{2}(\Exp_v(\tau_M) + 2)$, since with probability $1/2$ we jump to $M$ in 1 step, and otherwise we go to $u$ and then back to $v$, taking 2 steps.
This implies that $\Exp_v(\tau_M) = 3$ and hence $\Exp_{\pi|_S}(\tau_S^M)$.

\section{Proof of \texorpdfstring{$s$-$M$ and $S$-$M$}{s-M and S-M} commute times} \label{app:Cs-M}
In this appendix we prove \cref{claim:commute-time}.
It follows by generalizing \cite[Proposition 9.5]{levin2017MarkovChainsMixingTimes}, where the theorem is proven for the special case of $S$ and $M$ being singletons.
It builds on \textit{voltages}, which are dual to electric flows.
Any voltage is described by a function $h:X \to \R$ that is \textit{harmonic} on all nodes that are not sources or sinks, i.e.,
\[
h(u)
= \sum_{v \in X} P_{u,v} h(v)
\]
for every $u$ which is neither a source (that is, $u\neq s$) nor a sink (that is $u\not\in M$).
The quantities $C_{\pi|_S,M}$ and $C_{S,M}$ are described in \cref{def:flow}.

\lemmaSM*
\begin{proof}
First we prove the claim for a singleton $S=\{s\}$, in which case the claim becomes
	\[
	\Pr_s(\tau_M < \tau_s^+)
	= \frac{1}{C_{s,M} \pi_s}.
	\]
	
Define the \textit{boundary voltages} $h_B(s) = 0$ and $h_B(u \in M) = 1$.
By standard results \cite{bollobas2013modern}, this implies that a total current of magnitude $i = 1/R_{s,M}$ will flow from $s$ to $M$, and the resulting voltage can be uniquely described as the \emph{escape probability}
\[
h(u)
= \Pr_u(\tau_M < \tau_s),
\]
as shown in \cite[Proposition 9.1]{levin2017MarkovChainsMixingTimes} (see also \cite{doyle1984electric}).
Using that $h(s) = 0$ and $i_{u,v} = (h(v) - h(u))/w_{u,v}$ by Ohm's law, we can now rewrite
\begin{align*}
\Pr_s(\tau_M < \tau_s^+)
&= \sum_{v \in X\setminus\{s\}} P(s,v) \, \Pr_v(\tau_M < \tau_s) \\
&= \sum_{v \in X\setminus\{s\}} \frac{w_{s,v}}{w_s} (h(v) - h(s))
= \sum_{v \in X\setminus\{s\}} \frac{i_{s,v}}{w_s} = \frac{i}{w_s},
\end{align*}
with $i$ the total current.
Since $i = 1/R_{s,M} = W/C_{s,M}$ and $\pi_s = w_s/W$, this implies that $\Pr_s(\tau_M < \tau_s^+) = 1/(C_{s,M}\pi_s)$.

Now we reduce the general case to the singleton case. For this we consider the graph $G'$ where we replace $S$ by a single vertex $s'$, so that for $u,v\notin S$ we set $w'_{uv}:=w_{uv}$, $w'_{s'v}:=\sum_{s\in S}w_{sv}$, and $w'_{s's'}:=\sum_{s, r\in S}w_{sr}$. Clearly then $W'=W$, $\pi'(s')=\pi(S)$ and $R'_{s',M}=R_{S,M}$. The latter deserves a little explanation. One can see that in the optimal $S\to M$ flow for any two vertices $s_1,s_2\in S$ and $v\notin S$ we have $i_{s_1,v}/w_{s_1,v}=i_{s_2,v}/w_{s_2,v}$. Therefore, after merging the flows (currents) on the merged edges $(s_1,v)$, $(s_2,v)$ the dissipated power 
$$
	\frac{(i_{s_1,v}+i_{s_2,v})^2}{w_{s_1,v}+w_{s_2,v}}
	=(i_{s_1,v}+i_{s_2,v})\frac{i_{s_1,v}+i_{s_2,v}}{w_{s_1,v}+w_{s_2,v}}
	=(i_{s_1,v}+i_{s_2,v})\left(\frac{i_{s_1,v}}{w_{s_1,v}}=\frac{i_{s_2,v}}{w_{s_2,v}}\right)
	=\frac{i_{s_1,v}^2}{w_{s_1,v}}	+\frac{i_{s_2,v}^2}{w_{s_2,v}}	
$$
remains unchanged. So merging the flows / distributing flows proportionally to the edge weights gives a mapping between the optimal flows ($S\to M$ and $s'\to M$) without changing the objective.

Finally, observe that 
\begin{align*}
\Pr_{\pi|_S}(\tau_M < \tau_S^+)
&= \sum_{s\in S} \frac{\pi(s)}{\pi(S)}\sum_{v \in X\setminus S}  P(s,v) \, \Pr_v(\tau_M < \tau_S) \\
&= \sum_{s\in S} \frac{w_s}{w(S)}\sum_{v \in X\setminus S}  \frac{w_{s,v}}{w_s} \, \Pr_v(\tau_M < \tau_S) \\
&= \sum_{v \in X\setminus S} \sum_{s\in S} \frac{w_{s,v}}{w(S)} \, \Pr_v(\tau_M < \tau_S) \\
&= \sum_{v \in X\setminus S} P'(s',v) \, \Pr'_v(\tau_M < \tau_{s'}) \\
&= \Pr_{s'}(\tau_M < \tau_{s'}^+),
\end{align*}
and so
\[
	\Pr_{\pi|_S}(\tau_M < \tau_S^+)
	=\Pr_{s'}(\tau_M < \tau_{s'}^+)
	=\frac{1}{C'_{s',M} \pi(s')}
	=\frac{1}{C_{S,M} \pi(S)}. \qedhere
\]
\end{proof}

\subsection{Special case where $S = \{s\}$}
For the special case where $S$ is a singleton, this gives a tight characterization of the commute time.
This easily follows from combining Lemma \ref{lem:s-M} with the expression below.
This expression is proven in \cite[Proposition 2.3]{lovasz1993random} or \cite[Corollary 2.8]{aldous2002reversible} for the case where $M$ is a singleton, but it is easily extended to the more general case.

\begin{lemma} \label{lem:exp-returns}
Let $s$ be disjoint from $M$.
Then
\[
\Pr_s(\tau_M < \tau_s^+)
= \frac{1}{\Exp_s(\tau^s_M)\pi_s}.
\]
\end{lemma}
\begin{proof}
Let $q = \Pr_s(\tau_M < \tau_s^+)$.
Then by Kac's Lemma (Lemma~\ref{lem:kac}) we know that $\Exp_s(\tau_s^+) = 1/\pi_s$.
Necessarily, when starting from $s$, $\tau_s^+ \leq \tau^s_M$, and furthermore $\Pr_s(\tau_s^+ = \tau^s_M) = q$.
Now if $\tau_s^+ < \tau^s_M$, we know that the Markov chain is ``restarted'' at timestep $\tau_s^+$ (that is, it is distributed the same as when it started, namely, it is at $s$), and hence
\begin{align*}
\Exp_s(\tau^s_M - \tau_s^+)
&= q\Exp_s(\tau_M^s-\tau_s^+|\tau_M^s=\tau_s^+)+(1-q)\Exp_s(\tau_M^s-\tau_s^+|\tau_M^s>\tau_s^+)\\
\Exp_s(\tau_M^s)-\Exp_s(\tau_s^+)&=(1-q)\Exp_s(\tau^s_M).
\end{align*}
We can therefore rewrite $q = \Exp_s(\tau_s^+)/\Exp_s(\tau^s_M) = 1/(\Exp(\tau^s_M)\pi_s)$, proving the claim.
\end{proof}

Combining this lemma with our Lemma \ref{lem:s-M} shows that
\[
\Pr_s(\tau_M < \tau_s^+)
= \frac{1}{C_{s,M} \pi_s}
= \frac{1}{\Exp_s(\tau^s_M)\pi_s},
\]
and therefore $\Exp_s(\tau^s_M) = C_{s,M}$.
This generalizes the classic fact that $C_{s,t} = \Exp_s(\tau^s_t)$, as we mentioned in \cref{thm:commuteVertex}.

\end{document}